\documentclass[onecolumn]{IEEEtran}
\IEEEoverridecommandlockouts
\usepackage{cite}
\usepackage{amsmath,amssymb,amsfonts,amscd,amsxtra,amsthm}
\usepackage{algorithmic}
\usepackage{graphicx}
\usepackage{textcomp}
\usepackage{xcolor}
\usepackage{bm} 
\usepackage{graphicx}
\usepackage{tabularx}
\usepackage{caption}
\usepackage{esvect}
\usepackage[final]{microtype}
\usepackage{hyperref}

\newtheorem{thm}{Theorem}
\newtheorem{lem}{Lemma}
\newtheorem{cor}{Corollary}
\newtheorem{rem}{Remark}

\newcommand{\mc}[1]{\mathcal{#1}}
\newcommand{\mb}[1]{\mathbb{#1}}

\def\BibTeX{{\rm B\kern-.05em{\sc i\kern-.025em b}\kern-.08em
    T\kern-.1667em\lower.7ex\hbox{E}\kern-.125emX}}
\begin{document}

\title{On the Achievable Rate Region of  the $ K $-Receiver Broadcast Channels via Exhaustive Message Splitting\\
}

\author{Rui Tang, Songjie Xie, Youlong Wu
}

\maketitle

\begin{abstract}
This paper focuses on $ K $-receiver discrete-time memoryless broadcast channels (DM-BCs) with  private messages, where the transmitter wishes to convey $K$ private messages  to $K$  receivers respectively.  A general inner bound on the capacity region is proposed based on an  exhaustive message  splitting and a $K$-level  modified Marton's coding. The key idea is to  split every message into $ \sum_{j=1}^K {K\choose j} $ submessages each corresponding to a set of users who are assigned to recover them, and then send these submessages through codewords that are jointly typical with each other.   To guarantee the joint typicality among all transmitted codewords, a sufficient condition on the subcodebooks sizes is derived through a newly establishing  hierarchical covering lemma, which extends  the 2-level multivariate covering lemma to the $K$-level case including $(2^{K}-1)$   random variables with  more intricate dependence. As the number of auxiliary random variables and rate constraints both increase linearly with  $(2^{K}-1)$, 
the standard  Fourier-Motzkin elimination procedure becomes infeasible when $K$ is large. To tackle this problem, we obtain the final form of achievable rate region with a special observation of disjoint unions of sets that constitute the power set of $ \{1,\dots,K\}$.  The proposed achievable rate region allows  arbitrary input  probability mass functions (pmfs) and improves  over all previously known ones for $ K$-receiver ($K\geq 3$) BCs whose input pmfs should satisfy certain Markov chain(s).

\end{abstract}


\section{Introduction}

The 2-receiver discrete-time memoryless broadcast channels (DM-BCs) are first introduced by Cover ~\cite{Cover}, who proposed the prestigious superposition coding  that outperforms  the traditional time-division strategy. However, superposition coding is optimal only for certain categories of broadcast channels such as degraded, less noisy and more capable BCs~\cite{lsmc}. The best known inner bound on the capacity region of DM-BCs is achieved by Marton’s coding with message splitting~\cite{Marton}.  The key idea  is to split each source message into common and private parts, where the common part is encoded into a cloud-center codeword, and two private parts are encoded into two separate  codewords.  To enlarge the  achievable rate region, the submitted codewords are jointly typical which is 
guaranteed by a sufficient condition on the sizes of subcodebooks established by  the covering lemma \cite{{nit}}.

For the general $ K $-receiver DM-BCs, most previous work   mainly focused on superposition coding and message splitting (or merging)~\cite{sup1,sup2,sup3}  wherein they require certain Markov chains for auxiliary random variables (RVs). Especially, ~\cite{sup3} gives a general inner bound based on superposition coding and rate-splitting using notions from order theory and lattices wherein each receiver decodes its intended message (common or private) along with the partial interference designated to it through rate-splitting. The Marton's coding with 2-level superposition coding structure (consisting of one cloud-center codebook and $K$ satellite codebooks) can be easily constructed and analyzed by the
 multivariate  covering lemma and packing lemma \cite{{nit}}. 

In this paper, we consider $ K $-receiver DM-BCs including only private messages. To inherit the characteristics of superposition coding and combine them with Marton's coding, a general inner bound {is proposed based on an  \emph{exhaustive} message  splitting and a $K$-level modified Marton's coding. More specifically, every message  is split into $\sum_{j=1}^K {K\choose j} $ submessages,  with each corresponding to a set $\mathcal{S}$ belonging to the power set of $\{1,\ldots,K\}$. The submessages respecting to $\mathcal{S}$  are encoded into an {exclusive} codeword  and will be decoded by receiver $j$ if $j\in{\mathcal{S}}$.  To  obtain a potentially larger region, we enlarge the subcodebooks sizes and use a $K$-level Marton's coding to send all  codewords that are jointly typical with each other.

Note that there are mainly two challenges on establishing the final form of achievable  rate region.  The first one is how to derive  rate conditions such that all transmitted codewords are jointly typical. The solution  is related to the covering lemma. Unfortunately, the known multivariate covering lemma  only deals with  sequences which are generated conditionally independently \cite{nit},  while  in our scheme there are  $(2^{K}-1)$ RVs and the sequences are generated under more intricate dependence. We solve this problem by dividing  $(2^{K}-1)$ RVs into $K$ levels and process them 
hierarchically, leading to a new  lemma called   \emph{hierarchical covering lemma}.     The second challenge lies on how to apply  Fourier-Motzkin elimination procedure to obtain the final form of the inner bound.   As the number of auxiliary RVs and rate conditions increase linearly with  $(2^{K}-1)$, it's infeasible   using the  standard  Fourier-Motzkin elimination procedure, in particular when $K$ is large. To tackle this problem, we aggregate the rate of submessages based on a special observation of disjoint unions of sets that constitute the power set of $\{1,\ldots,K\}$, and finally establish   the final form of achievable  rate region.

In our proposed scheme, the exhaustive message splitting enables each user to decode different  fine particles of the source message according to its observed signal, and the $K$-level Marton's coding can  enlarge the inner bound by allowing  arbitrary dependence among the input RVs, rather than satisfying certain Markov chains as in \cite{sup1,sup2,sup3}. Compared to all previous work on inner bounds of $K$-receiver DM-BCs with only private messages, our inner bound includes their regions as special cases.


}


\subsection{Notations}
Throughout the paper, we use the notation in \cite{nit}. In particular, we let $$\mc{K}\triangleq\{1,\ldots,K\}$$ for a positive integer $K$.  We use  $\mathbb{P}(\mathcal{K})$ to denote  the power set of  $\mathcal{K}$, e.g., when $K=3$, $\mathbb{P}(\mathcal{K})=\{\{1\},\{2\},\{3\},\{1,2\},\{1,3\},$  $\{2,3\},\{1,2,3\}\}$.  
{For any non-empty set $\mc{T}\in \mathbb{P}(\mc{K})$ and $\mc{T}=\{\,  k_1,k_2,$  $\ldots,k_{|\mc{T}|}  \,\}$ where $ ( k_1,k_2,\ldots,k_{|\mc{T}|}) $ is in an increasing order, define the permutation set as
$\Pi_{\mc{T}}\triangleq\{(  \pi(k_1),\pi(k_2),\ldots,\pi(k_{|\mc{T}|})):\text{ for all distinct bijections}$  $\pi :\mc{K}\to \mc{K} \}$.
}

 Let   $(X_{\mathcal{I}_1},\ldots,X_{\mathcal{I}_{2^{K}}})$ be a tuple of $2^{K}$ RVs where $\mathcal{I}_i\in \mathbb{P}(\mc{K})$ for $i=1,\ldots,2^K$.   Given a set $\mc{I}=\big\{\mc{I}_1, \mc{I}_2,\ldots,\mc{I}_{|\mc{I}|}\big\}$ with $\mc{I}_1<\mc{I}_2<\cdots<\mc{I}_{|\mc{I}|}$ in dictionary order (e.g., $\{\{1\},\{2\},$  $\{1,2\}\}$ is in dictionary order),  the subtuple of RVs with indices from $\mc{I}$ is denoted by $X\big(\mc{I}\big)\triangleq(X_{\mathcal{I}_1},\ldots,X_{\mathcal{I}_{|\mathcal{I}|}})$, and the corresponding  realizations by $x\big(\mc{I}\big)\triangleq(x_{\mathcal{I}_1},\ldots,x_{\mathcal{I}_{|\mathcal{I}|}})$. Similarly,  given $2^K$ random vectors $(X^n_{\mathcal{I}_1},\ldots,X^n_{\mathcal{I}_{2^{K}}})$, let $X^n(\mc{I})=(X_{\mc{I}_j}^n:\mc{I}_j\in\mc{I})$.  Let $\bigsqcup_{\mc{I}_j\in\mc{I}}\mc{S}_{\mc{I}_j}$ be a disjoint union where $\mc{S}_{\mc{I}_j}\subseteq  \mathbb{P}(\mc{K})$ and  $\mc{S}_{\mc{I}_j}\cap\mc{S}_{\mc{I}_k}=\emptyset$ if $j\neq k$.

%




For  positive integers $k$ and $j$, we define  $x_k^j\triangleq(x_{k,1},$  $\ldots, x_{k,j})$, and $X_k^j\triangleq(X_{k,1},\ldots, X_{k,j})$.

{We use $ \delta(\epsilon) >0$ to denote a function of $ \epsilon $ that tends to zero as $ \epsilon\to0 $. When there are multiple functions $ \delta_{1}(\epsilon),\ldots,\delta_{k}(\epsilon) $, we denote them all by a generic function $  \delta(\epsilon) $ with the understanding that $ \delta(\epsilon)=\max\{\delta_{1}(\epsilon),\ldots,\delta_{k}(\epsilon)\} $.}

 {Let an all-one column vector $ (1, \ldots , 1) $ with a specified dimension be denoted by \textbf{1}.} Albeit with an abuse of notation, we use a string of elements in a singleton to represent the set with only one element, e.g., $\bm{R}_{123}=\bm{R}_{\{123\}} $, $ \textbf{M}_{123}=\textbf{M}_{\{123\}} $, $ \mathbb{A}(1)=\mathbb{A}(\{1\}) $, $ \mathbb{B}_{k}(123)=\mathbb{B}_{k}(\{123\}) $, etc.




\section{Channel Model}
Consider a $K$-receiver DM-BC with only private messages depicted in Fig. ~\ref{fig:BC}. The setup is characterized by a  input alphabet $\mathcal{X}$, $K$  output alphabets {$(\mathcal{Y}_k: k \in \mc{K})$}, and a collection of channel transition pmfs $p\big( y_1,\dots,y_K \big| x\big)$.  At time $i\in[1:n]$, the transmitter sends the channel input $x_i\in\mathcal{X}$,  receiver $k\in\mc{K}$ observes the output $y_{k,i}\in\mathcal{Y}_k $. 


The goal of the communication is that the transmitter conveys  private messages $M_k$  to receiver~$k$, for $k\in\mc{K}$, respectively. Each $M_k$ is independently and uniformly distributed over the set $\mathcal{M}_k\triangleq[1: 2^{nR_k}]$, where  $R_k$ denotes the communication rate  of receiver $k$.

The encoder maps the messages $(M_1,M_2,\ldots,M_K)$ to a sequence $x_i\in\mathcal{X}$: \begin{equation}X_i=f^{(n)}(M_1,M_2,\ldots,M_K),\end{equation} 
and receiver~$k\in\mc{K}$ uses channel outputs $y_k^n$ 
to estimate $\hat{M}_k$ as a guess of messages $M_k$:
\begin{equation}
\hat{M}_k=g^{(n)}_{k}(Y_k^n).
\end{equation}

A rate region $( R_k:k\in\mc{K} )$ is called achievable if for every blocklength $n$, there exists an encoding function $f^{(n)}$ and  $K$  decoding functions $g^{(n)}_1,\ldots,g^{(n)}_K$  such that the error probability
\begin{equation}P^{(n)}_e=\text{Pr}\big\{\hat{M}_k\neq M_k,\exists~k\in\mc{K}\big\}\end{equation}
tends to zero as the  blocklength $n$ tends to infinity. The closure of the set of achievable rate tuple $( R_k:k\in\mc{K} )$ is called the \textit{capacity region}. 

\begin{figure}[!t]
    \centering
    \includegraphics[width=0.45\textwidth]{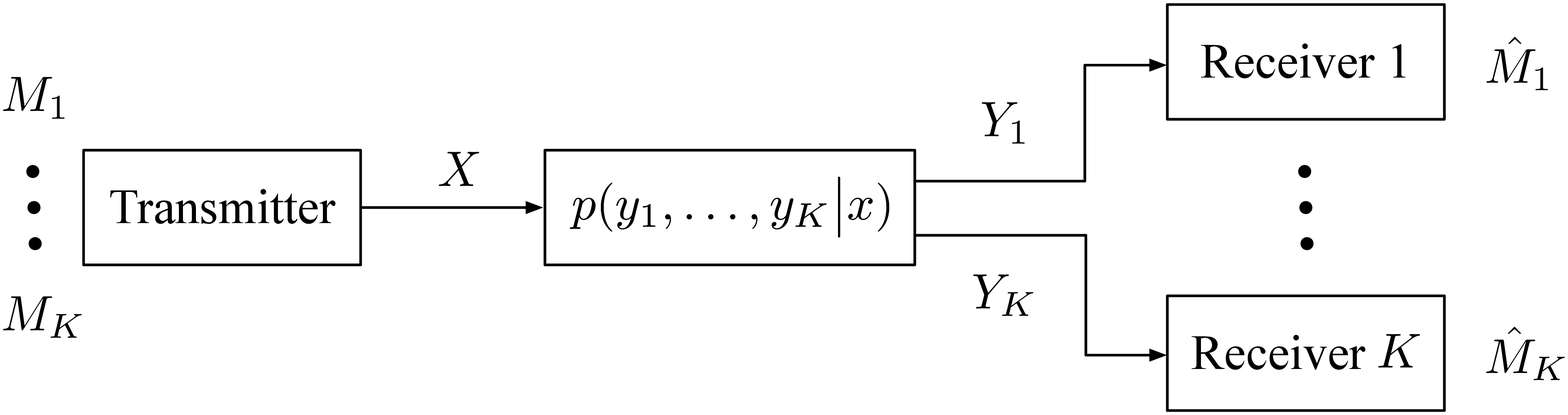}
    \caption{Broadcast Channel for $K$ receivers}
    \label{fig:BC}
    \vspace{-5mm}
\end{figure}

\section{Preliminary}\label{SecPre}
In this section, we present   decomposition of sets which will be used to in Section \ref{SecResults} and \ref{ACS}.

 Given a set 
 $\mc{S}\in \mathbb{P}(\mc{K})$, $k\in\mc{S}$ and $l$, $l'\in\mc{K}$ with $l\leq l'$, define 
 \begin{IEEEeqnarray*}{rCl}
 \mathbb{A}^{l:l'}\big( \mc{S} \big)&\triangleq&\big\{ \mc{S}'\in \mathbb{P}(\mc{K}):\,\mc{S}\subseteq\mc{S}' \,\text{and }\, |\mc{S}'|=l,l+1,\dots,l'\big\},\nonumber\\
 \mathbb{B}^{l:l'}_{k}(\mc{S}) &\triangleq& \{\,  \mc{S}'\subseteq \mc{S}:\,k\in\mc{S}', |\mc{S}'|=l,l+1,\ldots,l'  \,\}.
\end{IEEEeqnarray*}
 To simplify notations, we let   
 \begin{IEEEeqnarray*}{rClrClrCl}
 \mathbb{A}^{l:l'}&\triangleq&\mathbb{A}^{l:l'}\big( \emptyset  \big),&\mathbb{A}^{l}&\triangleq&\mathbb{A}^{l:l}\big( \emptyset  \big), &\mathbb{A}\big( \mc{S} \big) &\triangleq&\mathbb{A}^{1:K}\big( \mc{S} \big), \nonumber\\
\mathbb{A}&\triangleq&\mathbb{A}^{1:K}(\emptyset),~&\mathbb{B}_{k}(\mc{S}) &\triangleq&\mathbb{B}^{1:|\mc{S}|}_{k}(\mc{S}),~&\mathbb{B}^l_{k}(\mc{S}) &\triangleq& \mathbb{B}^{l:l}_{k}(\mc{S}).
\end{IEEEeqnarray*}

Next we show that
\begin{IEEEeqnarray}{rCl}\label{Expansion}
 \bigcup_{k\in\mc{T}} \mathbb{A}(k)  =\bigsqcup_{i=0}^{|\mc{T}|-1} \bm{\mathcal{B}}_\pi(i),
\end{IEEEeqnarray} \label{eqDeB}
 where for some $  \pi\in\Pi_{\mc{T}} $
 \begin{IEEEeqnarray}{rCl}
 \bm{\mathcal{B}}_\pi(0)&\triangleq&\mathbb{B}_{\pi(i+1)} \big(  \mc{K} \big),\\
 \bm{\mathcal{B}}_\pi(i)&\triangleq&\mathbb{B}_{\pi(i+1)} \big(  \mc{K}\setminus\{  \pi(1),\ldots,\pi(i) \}   \big).
\end{IEEEeqnarray}
 Note that  $\bigcup_{k\in\mc{T}} \mathbb{A}(i)$ denotes the subset of $\mathbb{P}(\mc{K})$ whose elements  contains at least one index belonging to $\mc{T}$, i.e., 
 $$\bigcup_{k\in\mc{T}} \mathbb{A}(k)=\{\mc{S}\in \mathbb{P}(\mc{K}):  \mc{T}\cap\mc{S}\neq \emptyset, |\mc{S}|=1,\ldots,K\},$$
  and  $\bm{\mathcal{B}}_\pi(i)$ contains all  sets  which include $\pi(i+1)$, and are subsets of  $\mc{K}\setminus\{ \pi(1),\ldots,\pi(i) \} $. {Also, we can have the following relation:
\begin{IEEEeqnarray}{rCl}
\bm{\mathcal{B}}_\pi(i)=\mb{A}(\{\pi(i+1)\}) \backslash (  \cup_{k=1}^{i}\mb{A}(\{\pi(i+1)\pi(k)\})) .
\end{IEEEeqnarray}}
 Eq. \eqref{Expansion} describes how we decompose  $\bigcup_{i\in\mc{T}} \mathbb{A}(i)$ into disjoint sets $ \bm{\mathcal{B}}_\pi(i)$, for $i\in[0:(|\mc{T}|-1)]$.

For example, consider the case $K=3$, $\mc{T}=\{1,2,3\}$ and $\pi=(2,1,3)$, we have 
\begin{IEEEeqnarray*}{rCl}
\mathbb{A}^{1:2}\big( \{2\} \big)&=&\{\{2\},\{1,2\},\{2,3\}\},\\
\bigcup_{k\in\mc{T}} \mathbb{A}(k) &=&\{\{1\},\{2,3\},\{1,2\},\{1,3\},\{2,3\},\{1,2,3\}\},\\
\mathbb{B}^{1:2}_{1}(\{1,2,3\})&=&\{\{1\},\{1,2\},\{1,3\}\},\\
\bm{\mathcal{B}}_\pi(i)|_{i=0} &=& \{\{2\},\{1,2\},\{2,3\},\{1,2,3\}\}, \\
\bm{\mathcal{B}}_\pi(i)|_{i=1} &=& \{\{1\},\{1,3\}\},
~ \bm{\mathcal{B}}_\pi(i)|_{i=2} = \{3\}.
\end{IEEEeqnarray*}
We can easily find that \eqref{Expansion} is satisfied.  Fig. \ref{fig:decomposition}  is given for illustration. To simplify the notation, an arbitrary collection of sets, e.g., $ \{1,12,123\} $, is recognized as $ \{\{1\},\{1,2\},\{1,2,3\}\} $

With the  definitions above, we can also obtain the following decomposition:
\begin{IEEEeqnarray}{rCl}\label{Expansion2}
 \bigcup_{k\in\mc{T}} \mathbb{A}^l(k)  =\bigsqcup_{i=0}^{|\mc{T}|-1} \bm{\mathcal{B}}_\pi^l(i),~{\bigcup_{k\in\mc{T}} \mathbb{A}(k) =\sum_{l=1}^K\bigsqcup_{i=0}^{|\mc{T}|-1} \bm{\mathcal{B}}_\pi^l(i),}\quad
 \end{IEEEeqnarray} 
 where 
 $\bm{\mathcal{B}}^{l}(i)\triangleq\mathbb{B}^l_{\pi(i+1)} \big(  \mc{K}\setminus\{ \pi(m) \}_{m=1}^i  \big)$. Here $ \bigcup_{k\in\mc{T}} \mathbb{A}^l(k) $ denotes the subset of $\mathbb{P}(\mc{K})$ whose element each has cardinality $l$ and contains at least one index belonging to $\mc{T}$, i.e., 
 $$\bigcup_{i\in\mc{T}} \mathbb{A}^l(i)=\{\mc{S}\in \mathbb{P}(\mc{K}):  \mc{T}\cap\mc{S}\neq \emptyset, |\mc{S}|=l\},$$ and 
 $\bm{\mathcal{B}}^l(i)$ contains all the sets with cardinalities  $l$,   containing $\pi(i+1)$, and being contained in $\mc{K}\setminus\{ \pi(m) \}_{m=1}^i $.

\begin{figure}[!t]
	\centering
	\includegraphics[width=0.45\textwidth]{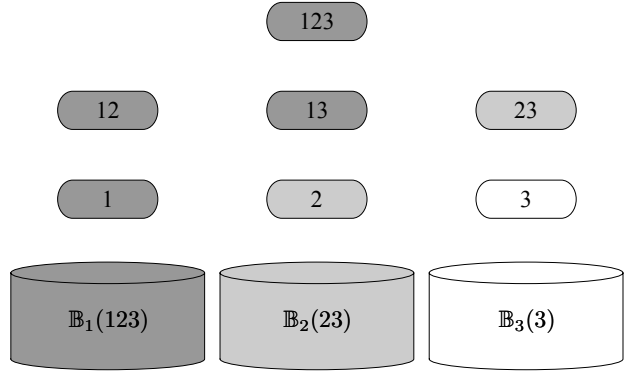}
	\caption{The illustration of the disjoint decomposition of $ \bigcup_{k\in \mc{T}} \mathbb{A}(k) $ in the case for $K=3$, $ \mc{T}=\{1,2,3\} $ and $ \pi=(1,2,3) $. The idea is that every time we pick a collection of sets descended from a set in the layer $ l $, the number of unpicked sets in layer $l-1$ is one, where layer $l$ contains all subsets of $\mc{K}$ with cardinality $l$. Therefore, $ \bigcup_{k\in[1:3]} \mathbb{A}(k) =  \mathbb{B}_{1}(123)\bigsqcup\mathbb{B}_{2}(23)\bigsqcup\mathbb{B}_{3}(3)$.}
	\label{fig:decomposition}
	\vspace{-5mm}
\end{figure}

\section{Main Results}\label{SecResults}

\begin{thm}\label{mainthm}
Using the notations presented in Section \ref{SecPre},	a rate region $\big( R_k:k\in\mc{K} \big)$ is achievable for the DM-BC $p\big( y_1,\dots,y_K \big| x\big)$ if for all  $\mc{T}\subseteq \mc{K}$ and $\pi \in\Pi_{\mc{T}}$,
	
	\begin{IEEEeqnarray}{rCl}\label{eq:rateregion}
	 &\sum_{k\in\mc{T}} R_{k}\nonumber
	\leq  \sum_{i=0}^{|\mc{T}|-1} \Big(  \sum_{\mc{S}'\in   \bm{\mathcal{B}}_\pi(i) } H\big(  U_{\mc{S}'}\mid U\big(\mathbb{A}(\mc{S}')\setminus \mc{S}'\big)\nonumber
\qquad \quad			\\
			&\qquad-  H \big(    U\big( \bm{\mathcal{B}}_\pi(i) \big) \mid  Y_{\pi(i+1)}  U\big( \mathbb{A}( \pi(i+1) )  \setminus  \bm{\mathcal{B}}_\pi(i)  \big)    \big) \Big)\nonumber\\
			&\qquad+ \sum_{l=1}^{K-1} \Big(   H\big(   U\big(\bigsqcup_{i=0}^{|\mc{T}|-1} \bm{\mathcal{B}}_\pi^{l}(i)\big)\mid U\big(\mathbb{A}^{(l+1):K}\big)   \big)
			\nonumber\\
			&\qquad\qquad- \sum_{\mc{S}\in\bigsqcup_{i=0}^{|\mc{T}|-1} \bm{\mathcal{B}}_\pi^{l}(i)} H\big(  U_\mc{S}\mid U\big(\mathbb{A}(S)\setminus \mc{S}\big)  \big)    \Big),
\end{IEEEeqnarray}
	for some  pmf $ p({U(\mathbb{A})})$ and a function $X=f(U(\mathbb{A}))$.

			\end{thm} 


\begin{proof}
	The achievable scheme is based  on an  exhaustive message  splitting and a modified Marton's coding. More specifically,  each message $M_k$, for $k\in\mc{K}$, is split into $\sum_{j=1}^K {K\choose j} $ submessages $M_k\triangleq\big( M_{k,\mc{S}}:\mc{S}\subseteq \mc{K},\, |\mc{S}|=1,2,\dots,K \big)$ with $M_{k,\mathcal{S}}\in\emptyset$ if $k\notin\mathcal{S}$.  The submessages $\textbf{M}_\mathcal{S}=(M_{k,\mathcal{S}}:k\in\mc{K})$ are encoded into a codeword $u^n_\mathcal{S}\sim p(u_{\mathcal{S}}|u(\mathbb{A}(\mc{S})\backslash\mc{S}))$ and will be decoded by receiver $j$ if $j\in{\mathcal{S}}$.  Then,  a modified Marton's coding is applied to send  $(2^{K}-1)$ codewords $(u^n_\mathcal{S}:\mathcal{S}\subseteq\mc{K},|\mathcal{S}|=1,\ldots,K)$ that are jointly typical with each other. 
  See more detailed proof in Section ~\ref{ACS}.		
\end{proof}

To ensure arbitrary input pmfs $ p({U(\mathbb{A})})$ in the inner bound above, we need to establish a sufficient condition on  the sizes of subcodebooks, which is related to the covering lemma. 
The  multivariate covering lemma in  \cite{nit} only consider sequences generated under simple dependence, e.g., $(u_1^n,\ldots,u_K^n)\sim p({u_1|u_0})p({u_1|u_0})\ldots p({u_K|u_0})$ given a sequence $u_0^n\sim P_{U_0}$. In our scheme there are  $(2^{K}-1)$ RVs, represented by $(U_\mathcal{S}:\mc{S}\in \mathbb{P}(\mc{K}),\mc{S}\neq \emptyset)$, and  each symbol of codeword $U^n_\mathcal{S}$ is generated conditionally independent according to $p(u_{\mathcal{S}}|u(\mathbb{A}(\mc{S})\backslash\mc{S}))$. At  first glance it seems overwhelming to derived the sufficient condition such that all transmitted codewords are jointly typical.  However, 
by  dividing all subcodebooks into $K$ levels and processing them recursively in a hierachical manner, we obtain the following new  lemma:

{
\begin{lem}[Hierarchical Covering Lemma]\label{hcl}
	Let $U(\mathbb{A})\thicksim p(u(\mathbb{A}))$ and 
	$\epsilon_{K}<\epsilon_{K-1}\cdots<\epsilon_{1}$. Let $U^{n}_{\mc{K}}\thicksim p\big(u^{n}_{\mc{K}}\big)$ be a random
	sequence with $\lim_{n\to\infty} P\big\{\,U^{n}_{\mc{K}}\in\mc{T}_{\epsilon_{K}}^{(n)}\,\}=1$.
	For each $l\in[1:K-1]$ and $\mc{S}\in\mathbb{A}^{l}$, let $U^{n}_\mc{S}(m_\mc{S})$, $m_S\in[1:2^{nr_\mc{S}}]$, be pairwise conditionally independent sequences, each distributed according to $\prod_{i=1}^{n}p_{U_{\mc{S}} | U(\mathbb{A}(\mc{S})\setminus\mc{S})}( u_{\mc{S},i}|u_{\mc{S}',i}:\mc{S}'\in(\mathbb{A}(\mc{S})\setminus\mc{S} ))$.
	Assume that $\big(U^{n}_\mc{S}(m_\mc{S}):m_S\in[1:2^{nr_\mc{S}}],\,\mc{S}\in\mathbb{A}^{l}\big)$ are mutually conditionally independent given $ (U^{n}_{\mc{S}}(m_{\mc{S}}):\mc{S}\in(\mathbb{A}(\mc{S})\setminus\mc{S} )$. Then there exists
	$\delta(\epsilon_{l})$ that tends to zero as $\epsilon_{l}\to$ 0 such that 
	\begin{align*}
		&\lim_{n\to\infty}P\Big\{\big(U^{n}_\mc{S}(m_\mc{S}):\mc{S}\in\mathbb{A}^{1:(K-1)},U^{n}_{\mc{K}}\big)\notin\mc{T}_{\epsilon_{1}}^{(n)}\\
		\,&\text{for all}~(m_\mc{S}:\mc{S}\in\mathbb{A}^{1:(K-1)})\in \!\!\!\! \prod_{\mc{S}\in\mathbb{A}^{1:(K-1)}} \!\!\!\![1:2^{nr_{\mc{S}}}] \Big\}=0,
	\end{align*}
	if
	\begin{equation}\label{eq:hclemma}
		\begin{aligned}
			\sum_{\mc{S}\in\mc{J}^{l}}r_{\mc{S}} > &\sum_{\mc{S}\in\mc{J}^{l}} H\Big(U_\mc{S}\mid U(\mathbb{A}(\mc{S})\setminus \mc{S})\Big)\\
			\mathrel{\phantom{>}}&- H\Big(U(\mc{J}^{l})\mid U\big(\mathbb{A}^{(l+1):K}\big)\Big)  + \delta(\epsilon_{l}),
		\end{aligned}
	\end{equation}
	for all $\mc{J}^{l}\subseteq \mathbb{A}^{l}$ and $l\in[1:(K-1)]$.
\end{lem}

\begin{proof}
	 See the  proof  in Appendix~\ref{prf:hcl}.
\end{proof}
}
\begin{lem}\label{combine}
	Consider the inequality condition in hierarchical covering lemma:
	\begin{align*}
		\sum_{\mc{S}\in\mc{J}^{l}}r_{\mc{S}} >&  \sum_{\mc{S}\in\mc{J}^{l}} H\big(U_\mc{S}\mid U(\mathbb{A}(\mc{S})\setminus \mc{S})\big)  \\
		&
		- H\big(U(\mc{J}^{l})\mid U\big(\mathbb{A}^{(l+1):K}\big)\big) + \delta(\epsilon).
	\end{align*}
	If the set $\mc{J}^l$ is split into $N\in [\,1:|\mc{J}^l|\,]$ disjoint pieces $ \big\{  \mc{J}^l_1,\mc{J}^l_2,\ldots,\mc{J}^l_N  \} $ satisfying
	\begin{IEEEeqnarray}{rCl}\label{Jsplit}
	\mc{J}^l=\bigcup_{i=1}^N \mc{J}^l_i,~\textnormal{with}~\mc{J}^l_i\cap \mc{J}^l_j=\emptyset, ~\textnormal{if}~i\neq j,
\end{IEEEeqnarray}  
	 then for all $ i\in[N]$, we have
	\begin{align*}
		\sum_{\mc{S}\in\mc{J}_i^{l}}r_{\mc{S}} > &
		\sum_{\mc{S}\in\mc{J}_i^{l}} H\big(U_\mc{S}\mid U(\mathbb{A}(\mc{S})\setminus \mc{S})\big) \\  	
		&- H\big(U(\mc{J}_i^{l})\mid U\big(\mathbb{A}^{(l+1):K}\big)\big) + \delta_{\mc{J}_i^{l}}(\epsilon).
	\end{align*}
	With simple justification, the sum of lower bounds for $\sum_{\mc{S}\in\mc{J}_i^{l}}r_{\mc{S}},\,\forall i\in [N]$ is not greater than the lower bound for $\sum_{\mc{S}\in\mc{J}^{l}}r_{\mc{S}}$, which means combinations of split inequalities will not induce a smaller region bounded by the overall inequalities that simultaneously restrain all variables, 
\end{lem}
\begin{proof}
{It's equivalent to prove that 
	\begin{IEEEeqnarray}{rCl}\label{upperbound}
	&& \sum_{\mc{S}\in\mc{J}^{l}} H\big(U_\mc{S}\mid U(\mathbb{A}(\mc{S})\setminus \mc{S})\big) \nonumber\\
				&&\quad - H\big(U(\mc{J}^{l})\mid U\big(\mathbb{A}^{(l+1):K}\big)\big)   \nonumber \\
				&&\stackrel{(a)}=\sum_{i=1}^N\sum_{\mc{S}\in\mc{J}_i^{l}} \Big(H\big(U_\mc{S}\mid U(\mathbb{A}(\mc{S})\setminus \mc{S})\big) \nonumber\nonumber\\
				&&\quad\quad - H\big(U(\mc{J}_i^{l})\mid U\big(\mathbb{A}^{(l+1):K}\big), U(\mc{J}_{1}^{l}),\ldots,U(\mc{J}_{i-1}^{l})\big)\Big)    \nonumber\\
	&&\stackrel{(b)}\geq	\sum_{i=1}^N \Big( \sum_{\mc{S}\in\mc{J}_i^{l}} H\big(U_\mc{S}\mid U(\mathbb{A}(\mc{S})\setminus \mc{S})\big)\nonumber\\
				&&\quad\quad\quad\quad- H\big(U(\mc{J}_i^{l})\mid U\big(\mathbb{A}^{(l+1):K}\big)\big)  \Big), \nonumber
\end{IEEEeqnarray}
where (a) follows by the chain rule of entropy and \eqref{Jsplit}; (b) holds because conditioning reduces entropy.
	}
\end{proof}


\begin{cor}\label{corfor3}
Define 
	\begin{IEEEeqnarray*}{rCl}
		&&I_{12}:=\min_{i\in\{1,2\}}\{I(U_{12} U_{123}; Y_i)\},\\
	&&I_{13}:=\min_{i\in\{1,3\}}\{I(U_{13} U_{123}; Y_i)\},\\
	&&I_{23}:=\min_{i\in\{2,3\}}\{I(U_{23} U_{123}; Y_i)\},
	\end{IEEEeqnarray*}
	and \begin{IEEEeqnarray*}{rCl}
		&&\Delta
		:=I(U_2;U_{13}|U_{12}U_{23}U_{123})+I(U_3;U_{12}|U_{13}U_{23}U_{123})\\
		&&+I(U_1;U_3|U_2 U_{12} U_{13}U_{23} U_{123} )+I(U_1;U_{23}|U_{12}U_{13}U_{123})\\
		&&+I(U_1;U_2|U_{12}U_{13}U_{23}U_{123})+I(U_2;U_3|U_{12}U_{13}U_{23}U_{123}).
\end{IEEEeqnarray*}

	A rate region $ (R_{1}, R_{2}, R_{3}) $ is achievable for 3-receiver DM-BC $ p(y_{1}, y_{2}, y_{3}|x) $ if 
	\begin{subequations}
	\begin{IEEEeqnarray}{rCl}\label{threeuser}
	R_1 ~&&< I(U_1 U_{12} U_{13} U_{123}; Y_1 )-I(U_1; U_{23}|U_{12} U_{13} U_{123}),\quad\\
		R_2 ~&&< I(U_2 U_{12} U_{23} U_{123}; Y_2 )-I(U_2; U_{13}|U_{12} U_{23} U_{123}),\\
		R_3 ~&&< I(U_3 U_{13} U_{23} U_{123}; Y_3 )-I(U_3; U_{12}|U_{13} U_{23} U_{123}),\quad\quad\\
		R_1~&&+~R_2 < I_{12}+I(U_1  U_{13}; Y_1| U_{12}U_{123})\nonumber\\
		&&\ \quad + I(U_2  U_{23}; Y_2| U_{12}U_{123})-I(U_2;U_{13}| U_{13} U_{23}U_{123})\nonumber\\
		&&\ \quad-I(U_{13};U_{23}| U_{12}U_{123})-I(U_1;U_{23}| U_{12} U_{13}U_{123})\nonumber\\
		&&\ \quad-I(U_1;U_2| U_{12} U_{13} U_{23}U_{123}),\label{eq12}
		\end{IEEEeqnarray}
		\begin{IEEEeqnarray}{rCl}
		R_1 ~&&+~R_3 < I_{13}+I(U_1  U_{12}; Y_1| U_{13}U_{123})\nonumber\\
		&&\ \quad+ I(U_3  U_{23}; Y_3| U_{13}U_{123})-I(U_{12};U_{23}| U_{13}U_{123})\nonumber\\
		&&\ \quad -I(U_3;U_{12}| U_{13} U_{23}U_{123}-I(U_1;U_{23}| U_{12} U_{13}U_{123})\nonumber\\
		&&\ \quad-I(U_1;U_3| U_{12} U_{13} U_{23}U_{123})),\\		
		R_2 ~&&+~R_3 < I_{23}+I(U_2  U_{12}; Y_2| U_{23}U_{123})\nonumber\\
		&&\  \quad+ I(U_3  U_{13}; Y_3| U_{23}U_{123})-I(U_{12};U_{13}| U_{23}U_{123})\nonumber\\
		&&\ \quad-I(U_3;U_{12}| U_{13} U_{23}U_{123})-I(U_2;U_{13}| U_{12} U_{13}U_{123})\nonumber\\
		&&\ \quad-I(U_2;U_3| U_{12} U_{13} U_{23}U_{123}),\\		
		R_1 ~&&+~R_2+R_3 <I_{12} + I(U_1 U_{13}; Y_1| U_{12}U_{123})\nonumber\\
		&&\ \quad+ I(U_2 U_{23}; Y_2| U_{12}U_{123})+I(U_3;Y_3| U_{13} U_{23}U_{123})\nonumber\\
		&&\ \quad-I(U_{13}U_{23}| U_{12}U_{123})- \Delta,
		\\
		R_1 ~&&+~R_2+R_3 < I_{13} + I(U_1 U_{12}; Y_1| U_{13}U_{123})\nonumber\\
		&&\ \quad+ I(U_3 U_{23}; Y_3| U_{13}U_{123})+I(U_2;Y_2| U_{12} U_{23}U_{123})\nonumber\\
		&&\ \quad-I(U_{12}U_{23}| U_{13}U_{123})- \Delta,
		\\
		R_1 ~&&+~R_2+R_3 <I_{23} + I(U_2 U_{12}; Y_2| U_{23}U_{123})\nonumber\\
		&&\ \quad+ I(U_3 U_{13}; Y_3| U_{23}U_{123})+I(U_1;Y_1| U_{12} U_{13}U_{123})\nonumber\\
		&&\ \quad-I(U_{12}U_{13}| U_{23}U_{123})- \Delta,
\end{IEEEeqnarray}
\end{subequations}
for some pmf $p(u_1u_2u_3u_{12}u_{13} u_{23}u_{123})$ and a function $x=f(u_1u_2u_3u_{12}u_{13} u_{23}u_{123})$.
	
\end{cor}
\begin{proof}
	The result directly comes from Theorem~\ref{mainthm} with $ K=3 $. 
	For example, for $\mc{T}=\{1,2\}$ and $ \pi=(1,2) $, in  \eqref{eq:rateregion}, then we have $\mathbb{A}^{K}=\{123\}$,	$\mathbb{A}^{2:K}=\{12,13,23,123\}$, and 
,\begin{IEEEeqnarray*}{rCl}
\bm{\mathcal{B}}_\pi(0)&=&\mathbb{B}_{1} \big(  123  \big)
=\{123,12,13,1\},\\~\bm{\mathcal{B}}_\pi(1)&=&\mathbb{B}_{2} \big(  23  \big)=\{23,2\},
\end{IEEEeqnarray*}
which satisfy the decomposition $\cup_{i\in\{1,2\}}\mathbb{A}(i)= \{1, 12,13,23,123\} =\bm{\mathcal{B}}_\pi(0)\sqcup \bm{\mathcal{B}}_\pi({1})$.  With  $ \pi=(1,2) $ and $ \pi=(2,1) $, we  obtain  upper bound for $ R_{1}+R_{2} $ in~\eqref{eq12}. The remaining rate constraints are acquired similarly.
\end{proof}

\begin{rem}[Comparison with superposition coding and Marton's coding]
	Our achievable region generalizes that introduced by standard Marton's coding and superposition coding, i.e., both of them are special cases of our general-form intended coding. Furthermore, the rate region in Theorem~\ref{mainthm} contains the regions resulted from the two aforementioned coding schemes.
	\begin{itemize}
		\item Our coding degenerates into Marton's coding by setting $ U_{\mc{S}}=const $, $ \forall \mc{S}\in \mathbb{A}^{2:(K-1)} $, which indicates our rate region contains that derived by Marton's coding;
		\item Since superposition coding is optimal for degraded/less noisy/more capable DM-BC, which is barely guaranteed in many cases, i.e., without knowing the concrete relation according to the Markov chain, our rate region contains that derived from superposition coding. 
	\end{itemize}
		In the future work, we will evaluate our rate region for some specific  DM-BCs to show that our coding scheme strictly improves previously known inner bounds.
\end{rem}


\section{Achievable coding scheme for Theorem~\ref{mainthm}}\label{ACS}
We first present our scheme for  3-receiver DM-BC as an illustration, and then extend it to general $K$-receiver DM-BC model for $K\geq 2$.

 {To simplify notations, we denote $ \mc{C}_{\mc{S}}(\textbf{m}_\mc{S}) $ by $ \mc{C}_{\mc{S}}(\textbf{m}) $ with cognition of the subscript of $ \textbf{m} $ from $ \mc{C}_{\mc{S}} $, similarly for $ U^{n}_{\mc{S}}(\textbf{m}_{\mc{S}},l_{\mc{S}}) $ denoted by $ U^{n}_{\mc{S}}(\textbf{m},l) $ and its realization $ u^{n}_{\mc{S}}(\textbf{m}_{\mc{S}},l_{\mc{S}}) $ denoted by $ u^{n}_{\mc{S}}(\textbf{m},l) $. 
}
\subsection{ Coding scheme for 3-receiver DM-BCs}\label{example}


\subsubsection{Rate splitting} 
Divide $M_j\in[1:2^{R_j}]$, $ j \in \{1,2,3\} $, into four independent messages $(M_{j,\mc{S}}\in[1:2^{R_{j,\mc{S}}}]$: $ \mc{S}\in \mathbb{A}(j) )$. 
Hence, $R_j = \sum\limits_{\mc{S}\in\mathbb{A}(j)}R_{j,\mc{S}}$. More precisely
\[
\begin{split}
	M_1 & = (M_{1,123}, M_{1,12}, M_{1,13}, M_{1,1}),
	\\
	M_2 &= (M_{2,123}, M_{2,12}, M_{2,23}, M_{2,2}),
	\\
	M_3 &= (M_{3,123}, M_{3,23}, M_{3,13}, M_{3,3}), 
	\\
	R_1 & = R_{1,123}+R_{1,12}+R_{1,13}+R_{1,1},
	\\
	R_2 &= R_{2,123}+R_{2,12}+R_{2,23}+R_{2,2},
	\\
	R_3 &= R_{3,123}+R_{3,23}+R_{3,13}+R_{3,3}.
\end{split}
\]
For convenience, let  $\bm{R}_{123}\triangleq R_{1,123}+R_{2,123}+R_{3,123}$, $\bm{R}_{12} \triangleq R_{1,12}+R_{2,12}$, $\bm{R}_{13}\triangleq R_{1,13}+R_{3,13}$, $\bm{R}_{23}\triangleq R_{2,23}+R_{3,23}$, $\bm{R}_{1}\triangleq R_{1,1}$, $\bm{R}_{2}\triangleq R_{2,2}$, $\bm{R}_{3}\triangleq R_{3,3}$ and $\textbf{M}_{123}\triangleq (M_{1,123}, M_{2,123}, M_{3,123})$, $\textbf{M}_{12}\triangleq (M_{1,12}, M_{2,12})$, $\textbf{M}_{23} \triangleq  (M_{2,23}, M_{3,23})$, $\textbf{M}_{13}\triangleq (M_{1,13}, M_{3,13})$, $\textbf{M}_{1}\triangleq M_{1,1}$, $\textbf{M}_{2}\triangleq M_{2,2}$, $\textbf{M}_{3}\triangleq M_{3,3}$.
\subsubsection{Codebook generation} 
Fix a pmf $p(u_{123}, u_{12}, u_{23}, u_{13}, $  $u_{1}, u_{2}, u_{3})$. Randomly and independently generate $2^{n\bm{R}_{123}}$ sequences $u_{123}^n(\textbf{m})$ each according to $\prod_{i = 1}^np_{U_{123}}(u_{123,i})$. 
For each $\textbf{m}_{12} \in \prod_{i=1}^{3}[1:2^{nR_{i,12}} ]$ generate a subcodebook $\mathcal{C}_{12}(\textbf{m})$ consisting of $2^{n(\tilde{\bm{R}}_{12}-\bm{R}_{12})}$ independent generated sequences $u_{12}^n(\textbf{m}, {l})$, ${l}_{12} \in [ 1 : 2^{n(\tilde{\bm{R}}_{12}-\bm{R}_{12})}]$ , each according to $\prod_{i=1}^n p_{U_{12}|U_{123}}(u_{12,i}|u_{123,i})$.
In the same way , for each $\textbf{m}_{23}$ and $\textbf{m}_{13}$, generate subcodebooks $\mathcal{C}_{23}(\textbf{m})$ and $\mathcal{C}_{13}(\textbf{m})$, respectively. For each $\textbf{m}_{1} \in [1: 2^{nR_{1,1}} ]$ generate a subcodebook $\mathcal{C}_{1}(\textbf{m})$ consisting of $2^{n(\tilde{\bm{R}}_{1}-\bm{R}_{1})}$ independent generated sequences $u_{1}^n(\textbf{m}, {l})$, ${l}_{1} \in [ 1 : 2^{n(\tilde{\bm{R}}_{1}-\bm{R}_{1})}]$ , each according to $\prod_{i=1}^n p_{U_{1}|U_{12}U_{13}U_{123}}(u_{1,i}|u_{12,i}u_{13,i}u_{123,i})$. For $ \textbf{m}_{2} $ and $ \textbf{m}_{3} $, generate corresponding subcodebooks similarly.

\subsubsection{Encoding}
For each $\textbf{m}_{\mc{S}}$, $ \mc{S}\in\{1,2,3,12,13, 23\} $, find an index tuple $({l}_{1}, {l}_{2},  {l}_{3},  {l}_{12},  {l}_{13},  {l}_{23} )$ such that $u_{\mc{S}}^n(\textbf{m}, {l})\in \mc{C}_{\mc{S}}(\textbf{m}), \mc{S} \in \{1,2,3,12,13, 23\}$, and 
\begin{align}
	(u_{\mc{S}}^n(\textbf{m}, l):\mc{S} \in \{1,2,3,12,13, 23\}) \in \mathcal{T}_{\epsilon'}^{(n)}
\end{align}
If more than one such index tuple can be found, choose an arbitrary one among those. If no such tuple exists, just choose $({l}_{1}, {l}_{2},  {l}_{3},  {l}_{12},  {l}_{13},  {l}_{23} )=(1,1,1,1,1,1)$. 
Then the transmitter  generates $x^n(\textbf{m}_{123}, \textbf{m}_{12}, \textbf{m}_{13}, \textbf{m}_{23},m_{1,1}, m_{2,2}, m_{3,3})$ as
\begin{align*}
	x_i= x(&u_{123,i}(\textbf{m}), u_{12,i}(\textbf{m}, l),u_{13,i}(\textbf{m}, l), u_{23,i}(\textbf{m}, l),\\
	&u_{1,i}(\textbf{m}, l), u_{2,i}(\textbf{m}, l)),u_{3,i}(\textbf{m}, l))),~\text{for}~i=1,\ldots,n.
\end{align*}

\subsubsection{Decoding}
Decoder $j=1$ declares that $(\hat{m}_{1,123},\hat{m}_{1,12},$  $ \hat{m}_{1,13}, \hat{m}_{1,1})$ is sent if it is the unique message such that 
\[
\begin{split}
	(u^n_{123}(\hat{\textbf{m}}), u^n_{\mc{S}}(\hat{\textbf{m}},l):\mc{S}\in\{1,12,13\}, y_1^n) \in \mathcal{T}_{\epsilon}^{(n)}.
\end{split}
\]

Similarly, receiver $j = 2,3$ uses joint typicality decoding to find the unique message 
$(\hat{m}_{2,123},\hat{m}_{2,12}, \hat{m}_{2,23}, \hat{m}_{2,2})$ and $ (\hat{m}_{3,123},\hat{m}_{3,13}, \hat{m}_{3,23}, \hat{m}_{3,3})$, respectively.

\subsubsection{Analysis of the probability of error}
Assume without loss of generality that 
\[ 
\begin{split}
	(\textbf{M}_{123}, \textbf{M}_{12},\textbf{M}_{13},\textbf{M}_{23}, \textbf{M}_1,  \textbf{M}_2, \textbf{M}_3) =(\textbf{1}, \textbf{1}, \textbf{1}, \textbf{1},  1,  1, 1)
\end{split}
\]is sent and let $L_{12}, L_{13},L_{23}, L_{1}, L_{2}, L_{3}$ be the index tuple of selected sequences 
\[
\begin{split}
	&\big(U_{123}^n(\textbf{1}),U_{12}^n(\textbf{1},L),U_{13}^n(\textbf{1},L), U_{23}^n(\textbf{1},L),\\
	&\quad U_1^n(L), U_2^n( L),U_3^n(L)\big)\\
	&\quad \in \mathcal{C}_{12}(\textbf{1})\times\mathcal{C}_{13}(\textbf{1})\times\mathcal{C}_{23}(\textbf{1})\times\mathcal{C}_{1}(1)\times\mathcal{C}_{2}(1)\times\mathcal{C}_{3}(1).
\end{split}
\]
{We note that the subcodebook $\mathcal{C}_\mc{S}(\textbf{m})$ consists of $2^{n(\tilde{\bm{R}}_\mc{S}-\bm{R}_\mc{S})}$ i.i.d. $U^n_{\mc{S}}(\textbf{m},l)$ sequences, $\forall \mc{S}\in \{{1}, {2},  {3},  {12},  {13},  {23} \}$. By the hierarchical covering lemma (with $r_\mc{S} = \tilde{\bm{R}}_\mc{S}-\bm{R}_\mc{S}$),  we obtain a set of constraints for $ \mc{S}\in\{12,13,23\}$
\begin{subequations}\label{covering01}
	\begin{align}
		{r}_{12}+{r}_{13} >&
		I(U_{12};U_{13}|U_{123}),\\
		{r}_{12}+{r}_{23} >& 
		I(U_{12}; U_{23}|U_{123}),\\
		{r}_{13}+{r}_{23} >& 
		I(U_{13}; U_{23}|U_{123}),\\
		{r}_{12}+{r}_{13}+{r}_{23}>&I(U_{12};U_{23}|U_{123})\nonumber\\&\hspace{-12ex}+I(U_{13};U_{23}|U_{123})+I(U_{12};U_{13}|U_{123},U_{23}).
	\end{align}
\end{subequations}
Again, using hierarchical covering lemma for $ \mc{S}\in\{1,2,3\} $, we attain the region for $ {r}_{1}, {r}_{2}, {r}_{3}$:
\begin{subequations}\label{covering02}
	\begin{align}
		&{r}_1 > I(U_1;U_{23}|U_{123}U_{12}U_{13}),
		\\
		&{r}_2>I(U_2;U_{13}|U_{123}U_{12}U_{23}),
		\\
		&{r}_3>I(U_3;U_{12}|U_{123}U_{13}U_{23}),\\
		&{r}_1+{r}_2>I(U_1;U_2|U_{123}U_{12}U_{13}U_{23})\nonumber\\
		&+I(U_1;U_{23}|U_{123}U_{12}U_{13})\!+\!I(U_2;U_{13}|U_{123}U_{23}U_{12}),\\
		&{r}_1+{r}_3>I(U_1;U_3|U_{123}U_{12}U_{13}U_{23})\nonumber\\
		&+ I(U_1;U_{23}|U_{123}U_{12}U_{13}) + I(U_3;U_{12}|U_{123}U_{13}U_{23}),\\
		&{r}_2+{r}_3>I(U_2;U_3|U_{123}U_{12}U_{13}U_{23})\nonumber\\
		&+ I(U_2;U_{13}|U_{123}U_{12}U_{23})+I(U_3;U_{12}|U_{123}U_{13}U_{23}),\\
		&{r}_1+{r}_2+{r}_3> I(U_1;U_2|U_{123}U_{12}U_{23}U_{13})\nonumber\\
		&+ I(U_2;U_3|U_{123}U_{12}U_{13}U_{23})+I(U_3;U_{12}|U_{123}U_{13}U_{23})\nonumber\\
		&+ I(U_1;U_{23}|U_{123}U_{12}U_{13})+I(U_2;U_{13}|U_{123}U_{23}U_{12})\nonumber\\
		&+ I(U_1;U_3|U_{123}U_{12},U_{23}U_{13}U_2). 
	\end{align}
\end{subequations}

By the symmetry of decoders, we first consider the average probability of error for decoder $1$. To analyze it, the Table~\ref{tab2} list all possible pmfs of message tuple $(U_{123}^n(\textbf{m}), U_{12}^n(\textbf{m}, l'), U_{13}^n(\textbf{m}, l'), U_1^n(\textbf{m}, l'), Y_1^n )$. 
\begin{IEEEeqnarray*}{rCl}
	&&\mc{E}_{1,0} = \{(U^n_{123}(\textbf{1}), U^n_{12}(\textbf{1},l), U^n_{13} (\textbf{1},l),U_1^n(\textbf{1},l), Y_1^n) \notin\mathcal{T}_{\epsilon}^{(n)}\},
	\\
	&&\mc{E}_{1,1} = \{(U^n_{123}(\textbf{1}), U^n_{12}(\textbf{1},l), U^n_{13} (\textbf{1},l), U^n_1(*), Y_1^n) \in \mathcal{T}_{\epsilon}^{(n)}\},
	\\
	&&\mc{E}_{1,13} = \{(U^n_{123}(\textbf{1}), U^n_{12}(\textbf{1},l), U^n_{13} (*), U^n_1(*), Y_1^n) \in \mathcal{T}_{\epsilon}^{(n)}\},
	\\
	&&\mc{E}_{1,12} = \{(U^n_{123}(\textbf{1}), U^n_{12}(*), U^n_{13} (\textbf{1},l), U^n_1(*), Y_1^n) \in \mathcal{T}_{\epsilon}^{(n)}\},
	\\
	&&\mc{E}_{1,12,13} = \{(U^n_{123}(\textbf{1}), U^n_{12}(*), U^n_{13} (*), U^n_1(*), Y_1^n) \in \mathcal{T}_{\epsilon}^{(n)}\},
	\\
	&&\mc{E}_{1,123} = \{(U^n_{123}(*), U^n_{12}(*), U^n_{13} (*), U^n_1(*), Y_1^n) \in\mathcal{T}_{\epsilon}^{(n)} \}.
\end{IEEEeqnarray*}

\renewcommand\arraystretch{1.4}
\begin{table}[t!]
	\centering
	\begin{tabular}{|p{0.8cm}|p{1.2cm}|p{1.2cm}|p{0.6cm}|p{2.3cm}|}
		\hline
		$\textbf{m}_{123}$ & $(\textbf{m}_{12}, l'_{12})$ & $(\textbf{m}_{13}, l'_{13})$ & $m_{1,1}$ & Joint pmf \\
		\hline
		$\textbf{1}_{123}$ & $(\textbf{1}_{12}, l_{12})$ &$(\textbf{1}_{13}, l_{13})$ & 1 & $p^*p(y^n_1|u_{123}^nu_{12}^nu_{13}^nu_1^n)$  \\ 
		\hline
		$\textbf{1}_{123}$&$(\textbf{1}_{12}, l_{12})$&$(\textbf{1}_{13}, l_{13})$& *  & $p^*p(y^n_1|u_{123}^n u_{12}^n u_{13}^n)$ \\
		\hline
		$\textbf{1}_{123} $& $(\textbf{1}_{12}, l_{12})$  & *  & * & $p^*p(y^n_1|u_{123}^n u_{12}^n)$\\
		\hline
		$\textbf{1}_{123}$ & *& $(\textbf{1}_{13}, l_{13})$ & * & $p^*p(y^n_1|u^n_{123}u^n_{13})$\\
		\hline
		$\textbf{1}_{123}$ &*&* &* & $p^*p(y^n_1|u^n_{123})$\\
		\hline
		*  &*  &*& *  &$p^*p(y^n_1)$\\
		\hline
		\multicolumn{5}{l}{* denote the message or index $\textbf{m}_s \neq \textbf{1}_s, l'_{s} \neq l_{s}, s\in \{123, 12, 13, 1\}$.} \\
		\multicolumn{5}{l}{And $p^*$ denote joint pmf $p(u_{123}^n, u_{12}^n, u_{13}^n, u_1^n)$.}
	\end{tabular}
	\caption{Table of joint pmfs of all possible error messages }
	\label{tab2}
\end{table}

By the conditional typicality lemma, the first term $P(\mc{E}_{1,0})$ tends to zero as $n \to \infty$ since that $(u_{\mc{S}}^n(\textbf{m}_{\mc{S}}, {l}_{\mc{S}}):\mc{S} \in \{1,2,3,12,13, 23\}) \in \mathcal{T}_{\epsilon'}^{(n)}$  and $\epsilon' < \epsilon$.

Considering the second term $P(\mc{E}_{1,1})$, note that $U_1^n(l) \sim \prod_{i=1}^nP_{U_1}(u_{1,i})$ is independent of $Y_1^n$ for every $U_{1}(l) \notin \mathcal{C}_1(1)$, then the variables $U_1^n(*)$ is independent of $Y^n_1$ on condition of $(U^n_{123}(\textbf{1}), U^n_{12}(\textbf{1},l), U^n_{13} (\textbf{1},l))$. By the packing lemma,  $P(\mc{E}_{1,1})$ tends to zero as $n \to \infty$ if 
\begin{align*}
	\tilde{\bm{R}}_1 &<
	 I(U_1;Y_1|U_{123} U_{12} U_{13}).
\end{align*}
Similarly, $P(\mc{E}_{1,12})$, $P(\mc{E}_{1,13})$, $P(\mc{E}_{1,12,13})$ and $P(\mc{E}_{1,123})$ tends to zero as $n \to \infty$ if
\begin{align*}
	\tilde{\bm{R}}_{12}&+\tilde{\bm{R}}_1 <\nonumber\\
	&I(U_{12} U_1;Y_1|U_{123}U_{13})+I(U_{12};U_{13}|U_{123}),
	\\
	\tilde{\bm{R}}_{13}&+\tilde{\bm{R}}_1 <\nonumber\\
	&I(U_{13} U_1;Y_1|U_{123} U_{12})+I(U_{12};U_{13}|U_{123}),
	\\
	\tilde{\bm{R}}_{12}&+\tilde{\bm{R}}_{13}+\tilde{\bm{R}}_1 < \nonumber\\
	&I(U_{12} U_{13} U_1;Y_1|U_{123})+I(U_{12};U_{13}|U_{123}),
	\\
	\bm{R}_{123}&+\tilde{\bm{R}}_{12}+\tilde{\bm{R}}_{13}+\tilde{\bm{R}}_1 < \nonumber\\
	&I(U_{123} U_{12} U_{13} U_1;Y_1)+I(U_{12};U_{13}|U_{123}).
\end{align*}
For receiver 2 and 3, the corresponding inequalities can be derived in a similar way.

From the lower bound of linear combinations of $r_\mc{S}$ and the upper bound of linear combinations of $\tilde{\bm{R}}_{\mc{S}}$ about receiver 1, 2 and 3, eliminating $\bm{R}_{123}$, $\bm{R}_{12}$, $\bm{R}_{13}$, $\bm{R}_{23}$ and $\bm{R}_{1}, \bm{R}_{2}, \bm{R}_{3}$ by our elimination procedure, an example given in Figure.~\ref{fig:decomposition}, yields the characterization~\eqref{threeuser}.

\subsection{Coding scheme for $K\geq 2$}

For each $k\in\mc{K}$, split  message $M_k\in[1:2^{nR_k}]$    into $\sum_{j=1}^K {K\choose j} $ submessages $M_k\triangleq\big( M_{k,\mc{S}}:\mc{S}\subseteq \mc{K},\, |\mc{S}|=1,2,\dots,K \big)$ with $M_{k,\mathcal{S}}\in[1:2^{R_{k,\mathcal{S}}}]$ and $R_{k,\mathcal{S}}=0$ if $k\notin\mathcal{S}$. 
Similarly to $K=3$ case, we use the following notations 
\begin{IEEEeqnarray}{rCl}\label{eqDeR}
\tilde{\bm{R}}_{\mc{S}}\triangleq\sum_{k=1}^K \tilde{{R}}_{k,\mc{S}},~{\bm{R}}_{\mc{S}}\triangleq\sum_{k=1}^K {{R}}_{k,\mc{S}}, ~r_{\mc{S}}\triangleq\tilde{\bm{R}}_{\mc{S}}-{\bm{R}}_{\mc{S}},
\end{IEEEeqnarray}
 and  $\textbf{M}_{\mc{S}}\triangleq\big( M_{1,\mc{S}},\dots,M_{K,\mc{S}}\big)$.

\subsubsection{Codebook generation} Fix a pmf $p(u(\mathbb{A}))$ and function $x(u(\mathbb{A}))$ and let $\tilde{\bm{R}}_{\mc{S}}\geq \bm{R}_{\mc{S}}$, $\forall \mc{S}\in \mathbb{A}$.  Randomly and independently generate $2^{n\bm{R}_{\mc{K}}}$ sequences $u_{\mc{K}}^{n}(\textbf{m})$, $ \textbf{m}_{\mc{K}}\in\prod_{i=1}^{K}[1:2^{nR_{i,\mc{K}}} ] $ each according to $\prod_{i=1}^{n} p_{U_{\mc{K}}}(u_{\mc{K},i}) $. For all $ \mc{S}\in\mathbb{A}^{1:(K-1)}$,  construct a subcodebook $ \mc{C}_{\mc{S}}(\textbf{m}) $ consisting of $ 2^{n(\tilde{\bm{R}}_{\mc{S}}-\bm{R}_{\mc{S}})} $ i.i.d. generated sequences $u^{n}_{\mc{S}}(\textbf{m},j)$, $(\textbf{m}_{\mc{S}}  ,  j_{\mc{S}})\in \prod_{i=1}^{K}[ 1:2^{nR_{i,\mc{S}}} ]\times[ 1 : 2^{n(\tilde{\bm{R}}_{\mc{S}}-\bm{R}_{\mc{S}})}]$, each according to $\prod_{i=1}^{n}p_{U_{\mc{S}} | U(\mathbb{A}(\mc{S})\setminus\mc{S})}( u_{\mc{S},i}|u_{\mc{S}',i}:\mc{S}'\in(\mathbb{A}(\mc{S})\setminus\mc{S} ))$. 


\subsubsection{Encoding.} 
For each $ \textbf{m}_{\mc{S}}\in  \prod_{i=1}^{K}[ 1:2^{nR_{i,\mc{S}}} ]$, $ \mc{S}\in\mathbb{A} $, find an index tuple $ (j_{\mc{S}}:\mc{S}\in\mathbb{A}) $ such that,
 $l\in[1:(K-1)]$, $u^{n}_{\mc{S}}(\textbf{m},j)\in \mc{C}_{\mc{S}}(\textbf{m}) $ and 
$$\big((u^{n}_\mc{S}(\textbf{m},j):\mc{S}\in\mathbb{A}), u^{n}_{\mc{K}}(\textbf{m})\big)\in\mc{T}_{\epsilon_1}^{(n)}.$$
If there is more than one such tuple, pick an arbitrary one among them. If no such tuple exists, pick $(j_{\mc{S}}:\mc{S}\in\mathbb{A})=\textbf{1}$. Then generate $x^n(\textbf{m}_{\mc{S}}:\mc{S}\in \mathbb{A})$ with 
\begin{equation*}
	x_i = x(u_{\mc{K},i}(\textbf{m}) ,\, u_{\mc{S},i}(\textbf{m},j):\forall \mc{S}\in\mathbb{A}),
\end{equation*}
for $i\in[n]$.

To send the message tuple $(m_1,\ldots,m_K)=(\textbf{m}_{\mc{S}}: \mc{S}\in\mathbb{A})$, transmit $x^{n}(\textbf{m}_{\mc{S}}: \mc{S}\in\mathbb{A})$.

\subsubsection{Decoding.} Let $\epsilon>\epsilon'$. Decoder $i$, $i\in\mc{K}$, declares that $ \hat{m}_{i} $ is sent if it is the unique message such that $$ \big(y_{i}^{n}, u^{n}_{\mc{S}}(\textbf{m},j):\mc{S}\in \mathbb{A}^{1:(K-1)}(i),u_{\mc{K}}^{n}(\textbf{m})\big)\in\mc{T}_{\epsilon}^{(n)},$$
 for some $(u^{n}_{\mc{S}}(\textbf{m},j):\mc{S}\in \mathbb{A}^{1:(K-1)}(i),u_{\mc{K}}^{n}(\textbf{m}))$ $\in$ $ \prod_{\mc{S}\in\mathbb{A}^{1:(K-1)}} \mc{C}_{\mc{S}}( \textbf{1} )\times \mc{C}_{\mc{K}}( \textbf{1})$.

\subsubsection{Analysis of the probability of error.} 
Assume without loss of generality that $(M_1,\ldots,M_K)=(\textbf{M}_{\mc{S}}:\forall \mc{S}\in\mathbb{A})=(1,\ldots,1)$ and let $ (J_{\mc{S}}: \mc{S}\in \mathbb{A}^{1:(K-1)}) ) $ be the index tuple of the chosen sequences $(U_{\mc{S}}:\mc{S}\in\mathbb{A}^{1:(K-1)})\in \prod_{\mc{S}\in\mathbb{A}^{1:(K-1)}} \mc{C}_{\mc{S}}( \textbf{1} )$. Then decoder $ i $ makes an error only if one or more of the following events occur:

	\begin{align*}\label{error_event}
		&\mc{E}_0=\big\{  	(U^{n}_\mc{S}(\textbf{m},j):\mc{S}\in\mathbb{A}^{l:(K-1)},U_{\mc{K}}^{n}(\textbf{m}))\\
		&\notin\mc{T}_{\epsilon'}^{(n)}\forall (U^{n}_{\mc{S}}:\mc{S}\in\mathbb{A}^{l})\in \prod_{\mc{S}\in\mathbb{A}^{l}} \mc{C}_{\mc{S}}( \textbf{1} ), \forall 	l\in[1:(K-1)]\big\},\\
		&\mc{E}_{i,1}=\big\{  \big(Y_{i}^{n}, U^{n}_{\mc{S}}(\textbf{M},J)\!:\!\mc{S}\in\mathbb{A}^{1:(K-1)}(i),U_{\mc{K}}^{n}(\textbf{m})\big)\notin\mc{T}_{\epsilon}^{(n)}  \big\},\\		
		&\mc{E}_{i,2}=\big\{  \big(Y_{i}^{n}, U^{n}_{\mc{S}}(\textbf{M},j):\mc{S}\in 	\mathbb{A}^{1:(K-1)}(i),U_{\mc{K}}^{n}(\textbf{m})\big)\\
		&\!\!\in\mc{T}_{\epsilon}^{(n)} \!\!\text{ for some }\!\!(U^{n}_{\mc{S}}(\textbf{M},j):\mc{S}\in 	\mathbb{A}^{1:(K-1)}(i)\big)\notin  \!\!\!\!\!\!\!\!\!\!\!\!\prod_{\mc{S}\in\mathbb{A}^{1:(K-1)}(i)}\!\!\!\!\!\!\!\!\!\!\! \mc{C}_{\mc{S}}( \textbf{1} ) \big\}.
	\end{align*}

Therefore, the probability of error for decoder $ i $ is upper bounded as
\begin{equation*}
	P(\mc{E}_i)\leq P(\mc{E}_{0})+P(\mc{E}^{c}_{0}\cap\mc{E}_{i,1})+P(\mc{E}_{i,2}).
\end{equation*}

To bound $P(\mc{E}_{0})$, we utilize the hierarchical covering lemma proposed before. Noticing that sequences among the subcodebooks $ \mc{C}_{\mc{S}}(\textbf{m}) $, $ \mc{S}\in\mathbb{A}^{l} $, $ l\in[1:(K-1)] $ are mutually conditionally independent, the usage of hierarchical covering lemma is straightforward with ${r}_{\mc{S}}=\tilde{\bm{R}}_{\mc{S}} -\bm{R}_{\mc{S}}$. Hence, $P(\mc{E}_0)$ tends to zero as $n\to\infty$ if ~\eqref{eq:hclemma} holds for all $l\in$  $[1:(K-1)]$ and $ \mc{J}^{l}\in\mathbb{A}^{l} $.

To bound $ P(\mc{E}^{c}_{0}\cap\mc{E}_{i,1}) $, by the conditional typicality lemma ~\cite{nit} it tends to zero as $n\to\infty$.

 To bound $ P(\mc{E}_{i,2}) $, noticing that $ U^{n}_{\mc{S}}(\textbf{M},j)\thicksim  \prod_{i=1}^{n}$  $p_{U_{\mc{S}} | U(\mathbb{A}(\mc{S})\setminus\mc{S})}( u_{\mc{S},i}|u_{\mc{S}',i}\!\!:\!\!\mc{S}'\in\mathbb{A}(\mc{S})\setminus\mc{S} )$ and $ Y^{n}_{i} $ are independent of $ U^{n}_{\mc{S}}(\textbf{M},j) \notin \mc{C}_{\mc{S}}( \textbf{1}) $,  $\forall \mc{S}\in \mathbb{A}^{1:(K-1)}(i)$, as well as $ U_{\mc{K}}(\textbf{M}) \notin \mc{C}_{\mc{K}}(\textbf{1})$. Furthermore, if, $ \forall \mc{S}\in \mathbb{A}(i) $,  $ U^{n}_{\mc{S}} \notin \mc{C}_{\mc{S}}(\textbf{1})$, then by the conditional coding distribution, for all $ \mc{S}'\in \mathbb{B}_{i}(\mc{S}) $, we have $ U_{\mc{S}'}(\textbf{M},j) \notin \mc{C}_{\mc{S}'}(\textbf{1}) $. Thus, $ \forall l\in\mc{K} $, $ \forall \mc{J}\subseteq   \mathbb{A}(i)$ by identifying $\bigcup_{\mc{S}\in\mc{J}}\mathbb{B}_{i}(\mc{S})$, we obtain inequalities with the help of packing lemma:
\begin{equation}\label{packinglemma}
	\begin{aligned}
		& \sum_{\mc{S}\in \bigcup_{\mc{S}'\in\mc{J}}\mathbb{B}_{i}(\mc{S}')   }	\!\!\!\tilde{\bm{R}}_{\mc{S}}
	<\! \!\sum_{\mc{S}'\in \bigcup_{\mc{S}\in\mc{J}}\mathbb{B}_{i}(\mc{S})} H\big(  U_{\mc{S}'}\mid U\big(\mathbb{A}(\mc{S}')\setminus \mc{S}'\big)\big)\\
		& \qquad
	-H\big(  U(\!\bigcup_{\mc{S}\in\mc{J}}\!\mathbb{B}_{i}(\mc{S})) \big | Y_{i}U(\mathbb{A}(i)\setminus\!\bigcup_{\mc{S}\in\mc{J}}\!\mathbb{B}_{i}(\mc{S}) )\big)  -\delta(\epsilon).
	\end{aligned}
\end{equation}
\subsubsection{Eliminating $(\tilde{R}_{k,\mc{S}},{r}_{\mc{S}}: k\in\mc{K}, \mc{S}\in\mb{A})$}{
Due to massive numbers of $\tilde{R}_{k,\mc{S}}$ and ${r}_{\mc{S}}$, using standard Fourier-Motzkin elimination to obtain the achievable rate region is  disastrous. We first present observations which help the elimination, and then find all valid constraints for $ \sum_{k\in\mc{T}} R_{k}  $ for all $ \mc{T}\subseteq\mc{K} $.\\
\textbf{Observation}: 
Let $\mb{S}(i,\mc{J})\triangleq\bigcup_{\mc{S}\in\mc{J}}\mathbb{B}_{i}(\mc{S})$, and  denote the right hand term of \eqref{packinglemma} by $I_{\mb{S}(i,\mc{J}) }$, thus \eqref{packinglemma} can be rewritten  as  below.
\begin{equation}\label{EqPackingRew}
	\begin{aligned}
		& \sum_{\mc{S}\in \mb{S}(i,\mc{J}) }\tilde{\bm{R}}_{\mc{S}}
	< I_{\mb{S}(i,\mc{J})}.
	\end{aligned}
\end{equation}
a) If there exists $i', i''$  and $\mc{J}',\mc{J}''$ such that 
\begin{subequations}\label{EqOb1}
\begin{IEEEeqnarray}{rCl}
\mb{S}(i,\mc{J})=\mb{S}(i',\mc{J}')\cup \mb{S}(i'',\mc{J}''),
\end{IEEEeqnarray}
 then 
\begin{IEEEeqnarray}{rCl}
 I_{\mb{S}(i,\mc{J})} \leq I_{\mb{S}(i',\mc{J})'}+I_{\mb{S}(i'',\mc{J}'')}.
\end{IEEEeqnarray}
\end{subequations}
b) If there exists $i'$  and $\mc{J}'$ such that $\mb{S}(i,\mc{J})\subseteq\mb{S}(i',\mc{J}'),$ then 
\begin{IEEEeqnarray}{rCl}\label{EqOb2}
 I_{\mb{S}(i,\mc{J})} \leq I_{\mb{S}(i',\mc{J})'}.
\end{IEEEeqnarray}

 From the rate-splitting procedure, we have $R_{i}= \sum_{\mc{S}\in \mathbb{A}(i)}R_{i,\mc{S}}$, for all $i\in\mc{K}$. Thus in order to find rate constraints for $\sum_{k\in\mc{T}} R_{k} $, we must find all valid constraints for $  \sum_{k\in\mc{T}}\sum_{\mc{S}\in\mathbb{A}(k)} R_{k,\mc{S}} $. By definitions in \eqref{eqDeR}, we have \begin{IEEEeqnarray}{rCl}\label{EqMinR}
 \sum_{k\in\mc{T}}R_k\stackrel{(a)}{\leq}\!\!\sum_{\mc{S}\in\cup_{k\in\mc{T}}\mathbb{A}(k)} \!\!\bm{R}_{\mc{S}}&=& \sum_{\mc{S}\in \bigcup_{k\in\mc{T}} \mathbb{A}(k)} \!\!\tilde{\bm{R}}_\mc{S}-\!\! \sum_{\mc{S}\in \bigcup_{k\in\mc{T}} \mathbb{A}(k) } \!\! r_\mc{S},\nonumber\\
\end{IEEEeqnarray}
where (a) is an equation if and only if $ R_{k,\mc{S}}=0 $ for all $ k\notin\mc{T} $.
 
 Consider  a permutation $\pi\in\Pi_{\mc{T}}$. For easy reading, rewrite the decomposition \eqref{Expansion}:
 \begin{IEEEeqnarray}{rCl}\label{ExpansionA}
 \bigcup_{k\in\mc{T}} \mathbb{A}(k)  =\bigsqcup_{i=0}^{|\mc{T}|-1} \bm{\mathcal{B}}_\pi(i).
\end{IEEEeqnarray}
 
 For $R_{\pi(1)}$, it must involve $R_{\pi(1),\mc{S}}:\mc{S}\in \mb{A}(\pi(1))$. Since $\mb{A}(\pi(1))=\bm{\mathcal{B}}_\pi(0)=\mb{S}(\pi(1), \mc{K})$ is the largest set with every element containing $\pi(1)$,  by observation in \eqref{EqOb1} and \eqref{EqOb2}, we have the only valid rate constraint $\sum_{\mc{S}\in \mb{A}(\pi(1) )} \tilde{\bm{R}}_\mc{S}\leq I_{ \mb{A}(\pi(1))}=I_{\bm{\mathcal{B}}_\pi(0)}$, to support $R_{\pi(1)}$. 
 
 For  $R_{\pi(2)}$, since $\sum_{\mc{S}\in \mb{A}(\pi(1) )} \tilde{\bm{R}}_\mc{S}$ already contains $\sum_{\mc{S}\in \mb{A}(\{\pi(1)\pi(2)\} )} {{R}}_{k,\mc{S}}$, thus we need to find the rate constraints for $\sum_{\mc{S}\in \mb{A}(\pi(2)) \backslash   \mb{A}(\{\pi(1)\pi(2)\} ) } {{R}}_{k,\mc{S}}$. Since $ \mb{A}(\pi(2)) \backslash   \mb{A}(\{\pi(1)\pi(2)\} )= \bm{\mathcal{B}}_\pi(1)=\mb{S}\big(\pi(2),\mc{K}\backslash\{\pi(1)\}\big)$, and $\bm{\mathcal{B}}_\pi(1)$  is second largest  sets with every element containing $\pi(2)$ and excluding $\pi(1)$, by observation  in (\ref{EqOb1}) and \eqref{EqOb2}, we  obtain the only effective rate constraints  
 $\sum_{\mc{S}\in  \bm{\mathcal{B}}_\pi(1)} \tilde{\bm{R}}_\mc{S}\leq I_{\bm{\mathcal{B}}_\pi(1)} $ to support $R_{\pi(2)}$. Following the similar steps and by \eqref{EqMinR}, we find all valid rate constraints under the permutation $\pi$ for all $R_{\pi(k)}$, $  k\in[1:|\mc{T}|] $: 
 \begin{IEEEeqnarray}{rCl}
\sum_{k\in\mc{T}} R_{k} \leq   \sum_{\mc{S}\in\bigsqcup_{i=0}^{|\mc{T}|-1}} I_{\bm{\mc{B}}_{\pi(i)}} - \sum_{\mc{S}\in \bigcup_{k\in\mc{T}} \mathbb{A}(k) } r_\mc{S}.
\end{IEEEeqnarray}

Thus, formally we have for all $ \mc{T} \subseteq \mc{K} $ and $ \pi\in \Pi_{\mc{T}} $:
\begin{equation*}
		\begin{aligned}
	\,\,\sum_{k\in\mc{T}} R_{k}&{\leq}    \sum_{\mc{S}\in\bigsqcup_{i=0}^{|\mc{T}|-1}} I_{\bm{\mc{B}}_{\pi(i)}} - \sum_{\mc{S}\in \bigcup_{k\in\mc{T}} \mathbb{A}(k) } r_\mc{S}\\
	&\stackrel{(a)}{=}    \sum_{\mc{S}\in\bigsqcup_{i=0}^{|\mc{T}|-1}} I_{\bm{\mc{B}}_{\pi(i)}}- \sum_{l=1}^{K-1} \sum_{\mc{S}\in \bigsqcup_{i=0}^{|\mc{T}|-1} \bm{\mathcal{B}}_\pi^{l}(i) } r_\mc{S}\\
			&\stackrel{(b)}< \sum_{i=0}^{|\mc{T}|-1} \Big(  \sum_{\mc{S}'\in   \bm{\mathcal{B}}_\pi(i) } H\big(  U_{\mc{S}'}\mid U\big(\mathbb{A}(\mc{S}')\setminus \mc{S}'\big)
			\\
			&\quad\,\,-  H \big(    U\big(\bm{\mathcal{B}}_\pi(i) \big)\! \mid  \!Y_{\pi(i\!+\!1)}  U\big( \mathbb{A}( \pi(i\!+\!1) ) \! \!\setminus \! \bm{\mathcal{B}}_\pi(i)  \big)  \!  \big)\! \Big)\quad\\
			&\qquad+ \sum_{l=1}^{K-1} \Big(   H\big(   U\big(\bigsqcup_{i=0}^{|\mc{T}|-1} \bm{\mathcal{B}}_\pi^{l}(i)\big)\mid U\big(\mathbb{A}^{(l+1):K}\big)   \big)
			\\
			&\qquad\qquad- \!\sum_{\mc{S}\in\bigsqcup_{i=0}^{|\mc{T}|-1} \bm{\mathcal{B}}_\pi^{l}(i)} H\big(  U_\mc{S}\mid U\big(\mathbb{A}(S)\setminus \mc{S}\big)  \big)    \Big),
		\end{aligned}
	\end{equation*}
where 
($a$) comes from the disjoint decomposition in \eqref{Expansion2}; ($b$) holds by letting $\mc{J}^l=\bigsqcup_{i=0}^{|\mc{T}|-1} \bm{\mathcal{B}}_\pi^{l}(i) $ in Lemma \ref{hcl} and by Lemma \ref{combine}.


 \section{Conclusion}
 In this paper, we propose new scheme for $ K $-receiver DM-BCs with  private messages based on    exhaustive message  splitting and   $K$-level Marton's coding. The achievable rate region allows arbitrary pmfs, and  improves  over all previously known ones for $ K\geq 3$ case whose input pmfs should satisfy certain Markov chain(s). A hierarchical  covering lemma is established which extends the $2$-level multivariate covering lemma to $K$-level case.
}
\appendices
\section{Proof of hierarchical covering lemma~\ref{hcl}}\label{prf:hcl}

Let
\begin{IEEEeqnarray}{rCl}
	\mathcal{M}_{\mathbb{A}^{l}} &\triangleq& 
	\prod_{{\mc{S}}\in {\mathbb{A}^{l}}} \mc{M}_{\mc{S}}=\prod_{\mc{S}\in\mathbb{A}^{l}} [1:2^{nr_{\mc{S}}}].
\end{IEEEeqnarray}

	{Divide the messages and codewords into $K$ levels, where at level $l\in[1:K]$, there are messages  sets	$\{\mc{M}_{\mc{S}}:\mc{S}\in\mathbb{A}^{l} \}$,
and codewords $$\big(U^{n}_{\mc{S}}(m_{\mc{S}}):{\mc{S}}\in {\mathbb{A}^{l}}, m_\mc{S}\in \mathcal{M}_{\mathbb{A}^{l}} \big).$$


	The proof of hierarchical covering lemma starts at level $ l=K-1 $, finding a message tuple $ (m_{\mc{S}}:\mc{S}\in\mathbb{A}^{K-1}) $ such that $ (u^{n}_\mc{S}(m_\mc{S}):\mc{S}\in\mathbb{A}^{K-1},u^{n}_{\mc{K}})\in \mc{T}_{\epsilon_{K-1}}^{(n)}$. Then with a descending order of $ l=K-2,K-1,\ldots,1$, we find a message tuple $ (m_{\mc{S}}:\mc{S}\in\mathbb{A}^{l}) $ such that $\big( (u^{n}_\mc{S}(m_\mc{S}):\mc{S}\in\mathbb{A}^{l}),u^{n}_{\mc{K}}\big)\in \mc{T}_{\epsilon_{l}}^{(n)}$. 
	
	 For each $l\in[1:K]$, we use $\mc{S}^l_i$ to denote the element of $\mathbb{A}^{l}$ for $i=1,\ldots,|\mathbb{A}^{l}|$, i.e.,
	 $$\mathbb{A}^{l}=\{\mc{S}^l_1,\ldots,\mc{S}^l_{|\mathbb{A}^{l}|}\}.$$
 We define, for all  $  l\in[1:K] $ , the random events
	
%

	\begin{IEEEeqnarray*}{rCl}
		\mc{E}_{l}&=&\Big\{\Big(\big( U^{n}_\mc{S}(m_\mc{S}):\mc{S}\in\mathbb{A}^{l:(K-1)}\big),U^{n}_{\mc{K}}\Big)\notin\mc{T}_{\epsilon_{l}}^{(n)},\\&&\quad\quad~\forall~
		(m_\mc{S}:\mc{S}\in\mathbb{A}^{l:(K-1)})\in \!\! \prod_{\mc{S}\in\mathbb{A}^{l:(K-1)}}\!\! [1:2^{nr_{\mc{S}}}] \Big\},\\
		\mc{A}_{l}&=&\big\{(m_{\mc{S}^l_1},\ldots, m_{\mc{S}^l_{|\mathbb{A}^{l}|}})\in\mathcal{M}_{\mathbb{A}^{l}} : \\
		&&\quad \big(U^n(m_{\mc{S}^l_1})\ldots U^n(m_{\mc{S}^l_{|\mathbb{A}^{l}|}}), \bm{u}^n(\mathbb{A}^{(l+1):K})\big)\in \mc{T}_{{\epsilon}_l}^{(n)}  \\
		&&\quad\quad \big| \bm{u}^n(\mathbb{A}^{(l+1):K}) \in \mc{T}_{\epsilon_{l+1}}^{(n)}\big\},\\
				\mc{Q}_{l}&=&  \big(u^n(m_{\mc{S}^j_1})\ldots u^n(m_{\mc{S}^j_{|\mathbb{A}^{j}|}})\!:\! j\in[(l+1)\!:\!K] \big)\in \mc{T}_{\epsilon_l}^{(n)},
\end{IEEEeqnarray*}
where 	\begin{align*}
	\bm{u}^n(\mathbb{A}^{(l+1):K})=&\big(u^n(m_{\mc{S}^j_1})\ldots u^n(m_{\mc{S}^j_{|\mathbb{A}^{j}|}}):\!\\
	&\! j\in[(l+1)\!:\!(K-1)] , u^{n}_{\mc{K}}\big)
\end{align*}.
	{Note that 	$\mc{E}_{l}$ denotes the event that there isn't any  tuple of sequences from level $l$ to $K$ such that they   are jointly typical. Correspondingly, $ \mc{E}_{l}^c$ denotes the event that there exists such tuple of sequences. $\mc{A}_{l}$ denotes the set of messages at level $l$ such that the corresponding codewords are jointly typical with a given tuple of jointly typical codewords from levels $l+1$ to $K$. $\mc{Q}_{l}$ denotes a condition that there exists message $(m_{\mc{S}^l_1},\ldots, m_{\mc{S}^l_{|\mathbb{A}^{l}|}})\in \mc{M}_{\mathbb{A}^{l}}$ at level $l+1,...,K$ such that the corresponding codewords are jointly typical. 
}	

Our goal is to find rate conditions such that  there exists a tuple of messages $(m_\mc{S}: \mc{S}\in\mathbb{A}^{1:(K-1)})$ satisfying  $\big(U^{n}_\mc{S}(m_\mc{S}):\mc{S}\in\mathbb{A}^{1:(K-1)},U^{n}_{\mc{K}}\big)\in\mc{T}_{\epsilon_{1}}^{(n)}$, which is equivalent to $ P\{  \mc{E}_{1}  \}\to 0$ as $n\to \infty$.

	The probability of $ \mc{E}_{1}$ can be upper bounded as 
	\begin{IEEEeqnarray*}{rCl}
	P\{  \mc{E}_{1}  \} 
	&=&P\{   \mc{E}_{1}\cap \mc{E}_{2}^{c}  \} + P\{    \mc{E}_{1}\cap \mc{E}_{2}  \}\\
	&\leq&P\{   \mc{E}_{1}| \mc{E}_{2}^{c}  \} + P\{ \mc{E}_{2}  \}.
\end{IEEEeqnarray*}
With recursively upper bounding $ P(\mc{E}_{l}) $, we obtain:
\begin{IEEEeqnarray}{rCl}\label{EqPro}
\label{errorprob}
	P\{\mc{E}_{1}\} &\leq &\sum_{l=1}^{K-1} P\{\mc{E}_{l}|\mc{E}_{l+1}^{c}\} + P\{\mc{E}_{K}\}\nonumber \\
	&=& \sum_{l=1}^{K-1} P\{|\mc{A}_l|=0\} + P\{\mc{E}_{K}\},
\end{IEEEeqnarray}
	where the last equality holds by  definitions of $\mc{A}_l$ and $\mc{E}_l$. Since $P(\mc{E}_{K}) $ tends to zero as $ n\to\infty $ by the property of typicality, and $K$ is a fixed channel parameter, in order to make $P\{\mc{E}_{1}\}$ tend to $0$, it's sufficient to let  $P\{|\mc{A}_l|=0\}$  tend  $0$ as $ n\to\infty $ for all   $l\in[1:K-1]$.

		To analyze~\eqref{errorprob}, $P\big\{|\mc{A}_{l}|=0\big\}  $ is upper bounded by a variation of Chebyshev inequality:
	\[P\big\{|\mc{A}_{l}|=0\big\}\!\leq\! P\big\{\big(|\mc{A}_{l}|-\text{E}|\mc{A}_{l}|\big)^2\!\leq\!\big(\text{E}|\mc{A}_{l}|\big)^2\big\}\!\leq\!\frac{\text{Var}\big(|\mc{A}_{l}|\big)}{\big(\text{E}|\mc{A}_{l}|\big)^2}.\]

	Using indicator random variables, $|\mc{A}_{l}|$ can be written as
	\begin{subequations}\label{eqAl}
	\begin{IEEEeqnarray}{rCl}
	|\mc{A}_{l}|&=&\sum_{j=1}^{|\mathbb{A}^l|}\sum_{m_{\mc{S}^l_j}\in \mc{M}_{\mc{S}^l_j}} \text{I}( m_{\mc{S}^l_1},\ldots, m_{\mc{S}^l_{|\mathbb{A}^{l}|}}) 
\end{IEEEeqnarray}
	where
	\begin{align}
		&\mathrel{\phantom{=}}\text{I}(m_{\mc{S}^l_1},\ldots, m_{\mc{S}^l_{|\mathbb{A}^{l}|}})=
		\begin{cases}
			1 & \text{if } \mc{Q}_{l} \text{ given } \mc{Q}_{l+1},\\
			0 & \text{otherwise,}
		\end{cases}
	\end{align}
	for each $(m_{\mc{S}^l_1},\ldots, m_{\mc{S}^l_{|\mathbb{A}^{l}|}})\in \mc{M}_{\mathbb{A}^{l}}$.\par
\end{subequations}
For all $\mc{J}^{l}\subseteq\mathbb{A}^{l}$, define
	\begin{IEEEeqnarray}{rCl}\label{eqJl}
		p_{\mc{J}^{l}} = P\big\{   &&  \big(U^n(m_{\mc{S}^l_1})\ldots U^n(m_{\mc{S}^l_{|\mathbb{A}^{l}|}}), \bm{u}^n(\mathbb{A}^{(l+1):K})\big)\in \mc{T}_{{\epsilon}_l}^{(n)},  \nonumber \\
		&&   \big(U^n(m'_{\mc{S}^l_1})\ldots U^n(m'_{\mc{S}^l_{|\mathbb{A}^{l}|}}), \bm{u}^n(\mathbb{A}^{(l+1):K})\big)\in \mc{T}_{{\epsilon}_l}^{(n)},  \nonumber \\
		&&  \text{ with } m_{\mc{S}^l_j}=m'_{\mc{S}^l_j}\text{ if }\mc{S}^l_j \in \mc{J}^l \text{ and vice versa}\nonumber\\ 
		&&\quad\quad \big| \bm{u}^n(\mathbb{A}^{l(+1):K}) \in \mc{T}_{\epsilon_{l+1}}^{(n)}\big\}.
\end{IEEEeqnarray}
From \eqref{eqAl} and definition of  $p_{\mc{J}^{l}}$ in \eqref{eqJl},  we have
	\begin{IEEEeqnarray*}{rCl}
	\text{E}\big(|\mc{A}_{l}|\big) &=&2^{n\sum_{\mc{S}\in\mathbb{A}^{l}}  r_\mc{S}} \cdot p_{\mathbb{A}^{l}},\\
		\text{E}\big(|\mc{A}_{l}|^2\big) &=&  \sum_{ m_{\mc{S}^l_1},\ldots, m_{\mc{S}^l_{\!|\mathbb{A}^{l}|}} } \sum_{ \substack{m'_{\mc{S}} \neq m_{\mc{S}}, \\ \forall \mc{S}\in \mathbb{A}^l\backslash \mc{J}^l} }    p_{\mc{J}^{l}}.\\
		&\leq& 2^{n\sum_{\mc{S}\in\mathbb{A}^{l}} r_\mc{S}}\cdot \sum_{\mc{J}^{l}\subseteq\mathbb{A}^{l}}2^{n\sum_{\mc{S}\in\mathbb{A}^l\setminus \mc{J}^{l}} r_{\mc{S}}}\cdot p_{\mc{J}^{l}}.
\end{IEEEeqnarray*}

From definition of $p_{\mc{J}^{l}}$ and independence of codewords, we have $(p_{\mc{A}_{l}})^2=p_\emptyset$.  Hence
	\[
	\text{Var}(|\mc{A}_{l}|) \leq 2^{n\sum_{\mc{S}\in\mathbb{A}^{l}} r_\mc{S}} \sum_{\mc{J}^{l}\subseteq \{\mathbb{A}^{l}\setminus\emptyset \}} 
 	2^{n\sum_{\mc{S}\in \mathbb{A}^l\setminus \mc{J}^{l}}  r_{\mc{S}}}\cdot p_{\mc{J}^{l}}.
	\]
 Next we compute the upper bound of   $p_{\mc{J}^{l}}$.

	\begin{IEEEeqnarray}{rCl}\label{eqPjl}
		p_{\mc{J}^{l}} \!=  \!&&\!\ P\big\{   \big(U^n(m_{\mc{S}^l_1})\ldots U^n(m_{\mc{S}^l_{|\mathbb{A}^{l}|}}), \bm{u}^n(\mathbb{A}^{(l+1):K})\big)\in \mc{T}_{{\epsilon}_l}^{(n)},  \nonumber \\
		&& \quad  \big(U^n(m'_{\mc{S}^l_1})\ldots U^n(m'_{\mc{S}^l_{|\mathbb{A}^{l}|}}), \bm{u}^n(\mathbb{A}^{(l+1):K})\big)\in \mc{T}_{{\epsilon}_l}^{(n)},   \nonumber\\
		&& \,\, \text{ with } m_{\mc{S}^l_j}=m'_{\mc{S}^l_j}\text{ if }\mc{S}^l_j \in \mc{J}_l \text{ and vice versa}\mid \mc{Q}_{l+1}\big\}\nonumber\\ 
		= &&\!\ P\big\{\! \big(U^n(m_{\mc{S}^l_1}\!)\ldots U^n(m_{\mc{S}^l_{\!\!|\mathbb{A}^{l}|}}\!),\! \bm{u}^n(\mathbb{A}^{l\!+\!1:K})\big)\!\in\! \mc{T}_{{\epsilon}_l}^{(n)}\!\mid\! \mc{Q}_{l\!+\!1}\big\}  \nonumber \\
		\!\!&&~\cdot\ P\Big\{   \big(U^n(m'_{\mc{S}^l_1})\ldots U^n(m'_{\mc{S}^l_{\!|\mathbb{A}^{l}|}}), \bm{u}^n(\mathbb{A}^{(l+1):K})\big)\in \mc{T}_{{\epsilon}_l}^{(n)}\nonumber\\
		&& ~\quad~~\mid\!\! \big(U^n(m_{\mc{S}^l_1}\!)\ldots U^n(m_{\mc{S}^l_{\!|\mathbb{A}^{l}|}}\!\!), \!\bm{u}^n(\mathbb{A}^{l\!+\!1:K})\big)\!\in \!\mc{T}_{{\epsilon}_l}^{(n)}, \!\mc{Q}_{l\!+\!1} \Big\} \nonumber \\
		\stackrel{(a)}{\leq} && \ 2^{n\big(H\big(U(\mathbb{A}^{l})|U(\mathbb{A}^{(l+1):K})\big)+\delta(\epsilon_{l}))\big)}\nonumber\\
		&&~\cdot\ 2^{-n\big( \sum_{\mc{S}\in\mathbb{A}^{l}} \!\!H\big(U_\mc{S}\mid U(\mathbb{A}(\mc{S})\setminus \mc{S}) -\delta(\epsilon_{l})\big)}\nonumber\\
		&&~\cdot\ 2^{n\big(H\big(U( \mathbb{A}^{l}\setminus\mc{J}^{l} )\mid U(  \mc{J}^{l}),U(\mathbb{A}^{(l+1):K}\big)\big) + \delta(\epsilon_{l})\big)}\nonumber\\
		&&~\cdot\ 2^{-n\big(\sum_{\mc{S}\in\mathbb{A}^{l}\setminus\mc{J}^{l}}H\big(U_\mc{S}\mid U(\mathbb{A}(\mc{S})\setminus \mc{S})\big) -  \delta(\epsilon_{l})\big)}
	\end{IEEEeqnarray}
where  ($a$)  follows by properties of joint typicality and the mutually conditional independence property, i.e., given $\bm{u}^{n}\big(\mathbb{A}^{(l+1):K}\big)$, $\big(U^{n}_\mc{S}(m'_\mc{S}):\mc{S}\in\mathbb{A}^{l}\setminus\mc{J}^{l}\big)$ and $\big(U^{n}_\mc{S}(m_\mc{S}):\mc{S}\in\mathbb{A}^{l}\setminus\mc{J}^{l}\big)$ are conditionally independent.
The lower bound of $p_{\mathbb{A}^{l}}$	can be derived as
	\begin{IEEEeqnarray}{rCl}\label{eqPal}
	p_{\mathbb{A}^{l}} &&= P\big\{   \big(U^n(m_{\mc{S}^l_1})\ldots U^n(m_{\mc{S}^l_{|\mathbb{A}^{l}|}}), \bm{u}^n(\mathbb{A}^{(l+1):K})\big)\in \mc{T}_{{\epsilon}_l}^{(n)}, \nonumber  \\
		&&\quad\quad \big| \bm{u}^n(\mathbb{A}^{(l+1):K}) \in \mc{T}_{\epsilon_{l+1}}^{(n)}\big\}, \nonumber \\ 
		&&\geq (1-\epsilon_l)  \cdot 2^{n \big(H\big(U(\mathbb{A}^{l})|U(\mathbb{A}^{(l+1):K})\big)- \delta(\epsilon_l)\big)\big)\, } \nonumber\\
		&&\quad \cdot\ 2^{-n\big( \sum_{\mc{S}\in\mathbb{A}^{l}} H\big(U_\mc{S}\mid U(\mathbb{A}(\mc{S})\setminus \mc{S})\big)+\delta(\epsilon_l)\big)}. 
\end{IEEEeqnarray}
where the last equality holds by the properties of jointly typicality.	
	From \eqref{eqPjl} and \eqref{eqPal}, we have 
\begin{IEEEeqnarray*}{rCl}
\frac{p_{\mc{J}^{l}}}{(p_{\mathbb{A}^{l}})^2}\leq&&\  2^{n\big(H\big(U(\mathbb{A}^{l})|U(\mathbb{A}^{(l+1):K})\big)+\delta(\epsilon_{l}))\big)}\\
&&~\cdot\ 2^{-n\big( \sum_{\mc{S}\in\mathbb{A}^{l}} H\big(U_\mc{S}\mid U(\mathbb{A}(\mc{S})\setminus \mc{S}) -\delta(\epsilon_{l})\big)}\\
		&&~\cdot\ 2^{n\big(H\big(U\big( \mathbb{A}^{l}\setminus\mc{J}^{l} \big)\mid U\big(  \mc{J}^{l}\big),U\big(\mathbb{A}^{(l+1):K}\big)\big) + \delta(\epsilon_{l})\big)}\\
		&&~\cdot\ (1-\epsilon_l)^{-2}  \cdot 2^{2n\big(H\big(U(\mathbb{A}^{l})|U(\mathbb{A}^{(l+1):K})\big)+\delta(\epsilon_{l}))\big)\, } \nonumber\\
		&&~\cdot\ 2^{-2n\big( \sum_{\mc{S}\in\mathbb{A}^{l}} H\big(U_\mc{S}\mid U(\mathbb{A}(\mc{S})\setminus \mc{S})\big)+\delta(\epsilon_l)}\big) \\
	\leq &&\  (1-\epsilon_l)^{-2} \cdot 2^{- nH(U(\mathbb{A}^{l})\mid U(\mathbb{A}^{(l+1):K}))}  \\
	&&~ \cdot\ 2^{n(\sum_{\mc{S}\in\mathbb{A}^{l}} H(U_\mc{S}\mid U(\mathbb{A}(\mc{S})\setminus \mc{S}))}\\
	&&~\cdot\ 2^{n( H(U( \mathbb{A}^{l}\setminus\mc{J}^{l} )\mid U( \mc{J}^{l})U(\mathbb{A}^{(l+1):K}))}\\
	&&~\cdot\ 2^{-  n\sum_{\mc{S}\in\mathbb{A}^{l}\setminus\mc{J}^{l}}H(U_\mc{S}\mid U(\mathbb{A}(\mc{S})\setminus \mc{S}))} \cdot 2^{n\delta(\epsilon)}.
	\end{IEEEeqnarray*}

Therefore, we finally obtain the desired inequality
\begin{align*}
	&\mathrel{\phantom{=}}\frac{\text{Var}(|\mc{A}_l|)}{(\text{E}(|\mc{A}_l|))^2}\\
	&\leq (1-\epsilon)^{-2} \cdot  2^{-n \sum_{\mc{S}\in\mathbb{A}^{l}} r_\mc{S}}  \cdot \! \! \sum_{\mc{J}^{l}\subseteq\mathbb{A}^{l}\setminus\emptyset}   \!\! \Big(2^{n\sum_{\mc{S}\in\mathbb{A}^l\setminus \mc{J}^{l}} r_{\mc{S}} } \cdot\! \frac{p_{\mc{J}^{l}}}{(p_{\mathbb{A}^{l}})^2}\Big) \\
	&\leq (1-\epsilon)^{-2}\sum_{\mc{J}^{l}\subseteq\mathbb{A}^{l}\setminus\emptyset}   \Big(  2^{-n \sum_{\mc{S}\in\mc{J}^{l}}  r_{\mc{S}}   }\cdot 2^{n\delta(\epsilon)}\\
	&\mathrel{\phantom{=}}\cdot\ 2^{nH(U( \mathbb{A}^{l}\setminus\mc{J}^{l} )\mid U( \mc{J}^{l})U(\mathbb{A}^{(l+1):K})) - H(U(\mathbb{A}^{l})\mid U(\mathbb{A}^{(l+1):K}))}\\
	&\mathrel{\phantom{=}}\cdot\ 2^{n ( \sum_{\mc{S}\in\mathbb{A}^{l}} H(U_\mc{S}\mid U(\mathbb{A}(\mc{S})\setminus \mc{S})) -  \sum_{\mc{S}\in\mathbb{A}^{l}\setminus\mc{J}^{l}}H(U_\mc{S}\mid U(\mathbb{A}(\mc{S})\setminus \mc{S}))    )  }\Big),
\end{align*}
which tends to zero as $n\to\infty$ if for all $ \mc{J}^{l}\subseteq\mathbb{A}^{l} $
\begin{IEEEeqnarray}{rCl}
\sum_{\mc{S}\in\mc{J}^{l}}r_{\mc{S}} 
> &&~ H(U( \mathbb{A}^{l}\setminus\mc{J}^{l} )\mid U(  \mc{J}^{l}),U(\mathbb{A}^{(l+1):K}))\nonumber \\
		&&- \ H(U(\mathbb{A}^{l})\mid U(\mathbb{A}^{(l+1):K})) \nonumber\\
	&&+\   \sum_{\mc{S}\in\mathbb{A}^{l}} H(U_\mc{S}\mid U(\mathbb{A}(\mc{S})\setminus \mc{S})) \nonumber\\
	&&- \  \sum_{\mc{S}\in\mathbb{A}^{l}\setminus\mc{J}^{l}}H(U_\mc{S}\mid U(\mathbb{A}(\mc{S})\setminus \mc{S})) 
	+ \delta(\epsilon)\nonumber\\
	= &&\  H(U( \mathbb{A}^{l}\setminus\mc{J}^{l} )\mid U( \mc{J}^{l}),U(\mathbb{A}^{(l+1):K}))\nonumber\\
	&&- \  H(U(\mathbb{A}^{l})\mid U(\mathbb{A}^{(l+1):K})) \nonumber\\
	&&+ \  \sum_{\mc{S}\in\mc{J}^{l}} H(U_\mc{S}\mid U(\mathbb{A}(\mc{S})\setminus \mc{S})) + \delta(\epsilon)\nonumber\\
	= &&~\sum_{\mc{S}\in\mc{J}^{l}} H(U_\mc{S}\mid  U(\mathbb{A}(\mc{S})\setminus \mc{S})) \nonumber \\
	&&-\  H(U(\mc{J}^{l})\mid U(\mathbb{A}^{(l+1):K}))  + \delta(\epsilon).\label{coveringcondition}
\end{IEEEeqnarray}
Therefore, from \eqref{EqPro}, $ P\{\mc{E}_{1}\} $ tends to zero as $ n\to\infty $ if conditions~\eqref{coveringcondition} hold for all $ l\in[1:(K-1)] $ .}


\end{document}